\documentclass[letterpaper,11pt]{article}
\pdfoutput=1 
\usepackage[utf8]{inputenc}

\usepackage[margin=1in]{geometry}

\usepackage{amsmath,amsthm,amsfonts,amssymb,float,bm,bbm}

\usepackage{hyperref}
\usepackage[svgnames]{xcolor}
\hypersetup{colorlinks={true},urlcolor={blue},linkcolor={DarkBlue},citecolor=[named]{DarkGreen}}

\usepackage{microtype}
\usepackage[capitalise,nameinlink,noabbrev]{cleveref}

\usepackage{xspace}
\usepackage{doi}
\usepackage{graphicx}
\usepackage{mathtools}
\usepackage{subcaption}
\usepackage{todonotes}
\usepackage{authblk}

\allowdisplaybreaks

\newtheorem{informaltheorem}{Informal Theorem}
\newtheorem{theorem}{Theorem}[section]
\newtheorem{lemma}[theorem]{Lemma}
\newtheorem{corollary}[theorem]{Corollary}
\newtheorem{claim}{Claim}
\Crefname{claim}{Claim}{Claims}
\Crefname{informaltheorem}{Informal Theorem}{Informal Theorems}

\newtheorem{openproblem}{Open Problem}

\theoremstyle{definition}
\newtheorem{definition}{Definition}[section]

\def\fp/{\textup{\textsf{FP}}}
\def\p/{\textup{\textsf{P}}}
\def\np/{\textup{\textsf{NP}}}
\def\conp/{\textup{\textsf{co-NP}}}
\def\fnp/{\textup{\textsf{FNP}}}
\def\tfnp/{\textup{\textsf{TFNP}}}
\def\ptfnp/{\textup{\textsf{PTFNP}}}
\def\ppa/{\textup{\textsf{PPA}}}
\def\ppad/{\textup{\textsf{PPAD}}}
\def\ppads/{\textup{\textsf{PPADS}}}
\def\ppp/{\textup{\textsf{PPP}}}
\def\pwpp/{\textup{\textsf{PWPP}}}
\def\pls/{\textup{\textsf{PLS}}}
\def\cls/{\textup{\textsf{CLS}}}
\def\ppadpls/{\textup{$\textsf{PPAD} \cap \textsf{PLS}$}}
\def\ppapls/{\textup{$\textsf{PPA} \cap \textsf{PLS}$}}
\def\eopl/{\textup{\textsf{EOPL}}}
\def\sopl/{\textup{\textsf{SOPL}}}
\def\ueopl/{\textup{\textsf{UEOPL}}}
\def\fixp/{\textup{\textsf{FIXP}}}
\def\linearfixp/{\textup{\textsf{Linear-FIXP}}}
\def\pspace/{\textup{\textsf{PSPACE}}}

\newcommand{\ba}{\bm{a}}

\newcommand{\bs}{\bm{s}}
\newcommand{\bu}{\bm{u}}
\newcommand{\bx}{\bm{x}}
\newcommand{\by}{\bm{y}}
\newcommand{\hx}{\widehat{x}}
\newcommand{\hf}{\widehat{f}}
\newcommand{\hh}{\widehat{h}}

\newcommand{\bhx}{\bm{\widehat{x}}}

\DeclareMathOperator{\indicator}{\mathbbm{1}}

\newcommand{\sz}{\mathsf{size}}
\newcommand{\poly}{\mathsf{poly}}

\newcommand{\bR}{\mathbb{R}}
\newcommand{\bZ}{\mathbb{Z}}

\newcommand{\cB}{\mathcal{B}}

\renewcommand{\epsilon}{\varepsilon}
\newcommand{\eps}{\varepsilon}

\title{Computing Equilibrium Points of Electrostatic Potentials}

\author{Abheek Ghosh, Paul W.\ Goldberg, Alexandros Hollender}
\affil{University of Oxford, UK}
\date{}

\begin{document}

\maketitle

\begin{abstract}
We study the computation of {\em equilibrium points} of electrostatic potentials: locations in space where the electrostatic force arising from a collection of charged particles vanishes. This is a novel scenario of optimization in which solutions are guaranteed to exist due to a nonconstructive argument, but gradient descent is unreliable due to the presence of singularities.

We present an algorithm based on piecewise approximation of the potential function by Taylor series. The main insight is to divide the domain into a grid with variable coarseness, where grid cells are exponentially smaller in regions where the function changes rapidly compared to regions where it changes slowly. Our algorithm finds approximate equilibrium points in time poly-logarithmic in the approximation parameter, but these points are not guaranteed to be close to exact solutions. Nevertheless, we show that such points can be computed efficiently under a mild assumption that we call ``strong non-degeneracy''. We complement these algorithmic results by studying a generalization of this problem and showing that it is \cls/-hard and in \ppad/, leaving its precise classification as an intriguing open problem.
\end{abstract}

\section{Introduction}

Suppose that we have a finite set $P$ of electrically charged particles, immobile and located at given points in space. These particles give rise to an electrostatic potential: a function $V$ that maps any point $x$ in space to the amount of energy needed to move a unit positive charge from a point at infinity to $x$. The negated gradient of $V$ constitutes the force acting on a positively charged particle at $x$: it is attracted to negatively charged particles in $P$ and repelled by positively charged ones. Each $p\in P$ exerts a force at $x$ proportional to $r^{-2}$, $r$ being the distance between $x$ and $p$. In this paper, we study the computation of {\em equilibrium points} where the gradient of $V$ is zero, meaning that the forces exerted by members of $P$ cancel each other out. 

For simplicity, let us for now consider the setting where all particles are positively charged.\footnote{Incidentally, the setting where all particles are negatively charged corresponds to the setting of a gravitational field. Furthermore, note that an equilibrium point for a given instance of the problem remains an equilibrium point if the sign of every charged particle is flipped. So, the setting where all particles are positively charged also corresponds to the setting of a gravitational field for the purpose of studying equilibrium points.} If $P$ consists of a single particle, then no equilibrium point exists, as any (mobile) particle would just be pushed away from the single particle. However, if $P$ contains at least two particles, an equilibrium point is always guaranteed to exist! This is usually proved  by using Morse theory.\footnote{Morse theory~\cite[Chapter 32]{morse2014critical}, which is traditionally used to prove existence in this setting, only guarantees existence for \emph{generic} instances, i.e., when the positions of the charged particles lie outside a set of measure zero. Nevertheless, an equilibrium point is guaranteed to exist even without this genericity assumption; we provide a proof sketch in \cref{app:generalized-potentials}.} Importantly, the proof of existence is nonconstructive and thus does not yield an efficient algorithm for finding such an equilibrium point.

More generally, although the three-dimensional setting corresponds to the physical world, this problem can be considered in $d$-dimensional space for any $d \geq 2$, and one can also consider more general potentials, e.g., proportional to $r^{-k}$ for some $k$. The aforementioned nonconstructive existence result continues to hold and in this paper we initiate the study of conditions under which approximate equilibrium points can be computed efficiently.

\paragraph{\bf Contribution.}
We study the numerical computation of equilibrium points of electrostatic potentials: given locations of charged particles, and the sign and magnitude of their charges, we want to find approximate equilibrium points (where the force is within some given $\varepsilon$ of zero). For this problem, we present an algorithm that runs in polynomial time when the dimension $d$ of the space is constant.

\begin{informaltheorem}\label{infthm:weak-approx}
When the dimension $d$ is constant, an $\eps$-approximate equilibrium point of an electrostatic potential can be computed in time polynomial in the input and $\log(1/\eps)$.
\end{informaltheorem}

In particular, our algorithm is efficient in the physical setting where $d=3$ and even when $\eps$ is inverse-exponential.\footnote{The attentive reader might have observed that it is trivial to find a point where the force has norm less than $\eps$: it suffices to pick a point that is sufficiently far away from our set of particles $P$. In order to be able to exclude these artificial solutions, our algorithm also takes as input the description of a bounded domain, and only looks for solutions inside that domain. Furthermore, it is possible to pick a large enough domain that is guaranteed to contain all exact solutions (and thus, in particular, at least one solution).} Furthermore, it can handle the presence of both positive and negative charges.

Our algorithm has a conceptual link with the {\em trust region} approach in optimization \cite{ConnGT00}, in that we divide up the domain into regions within which we approximate the objective function with a more tractable kind of function (a polynomial). The difference is that here, we exploit our knowledge of the structure of the objective function so as to identify regions within which we have a prior guarantee that the approximation achieves the desired precision everywhere in the region. Our algorithm extends to a class of functions with ``well-behaved'' derivatives, where this subdivision of the domain remains efficient.

One important caveat of our algorithm is that it outputs \emph{weak} approximations of equilibrium points, namely points where the force is approximately zero. Such points are not guaranteed to be close to exact equilibrium points (where the force is exactly zero). Ideally, we would also like to be able to compute \emph{strong} approximations of equilibrium points, namely points that are actually close to exact equilibrium points.

\begin{openproblem}
Is there an efficient algorithm that uses algebraic methods in a more direct manner, avoiding the use of numerical approximation?
\end{openproblem}

The hope would be that such an algorithm would be able to output strong approximate solutions, instead of weak solutions. Although we leave this question open, we show that our numerical approach can be used to compute strong approximate solutions under an additional assumption that we call \emph{strong non-degeneracy}. This condition requires that at any equilibrium point, the determinant of the Hessian of the potential is bounded away from zero.

\begin{informaltheorem}\label{infthm:strong-approx}
When the dimension $d$ is constant and the instance is $\delta$-strongly non-degenerate, a point $\eps$-close to an exact equilibrium point can be computed in time polynomial in the input, $\log(1/\eps)$, and $\log(1/\delta)$.
\end{informaltheorem}

Note that we cannot hope to compute equilibrium points exactly, as these might be irrational.\footnote{For example, consider two positive charges with magnitudes 1 and 2, one unit apart, and an inverse-square force power law: the equilibrium point between them is $\sqrt{2}-1$ distant from the charge with value 1.} As with our first algorithm, the algorithm is only efficient when the dimension $d$ is fixed. Although this is sufficient for the physical setting, the electrostatic potential in high dimension is relevant to generative models in machine learning. Perhaps the main question left open by our work is thus the following.

\begin{openproblem}
Is there an efficient algorithm for computing weak (or strong) approximate equilibrium points of electrostatic potentials when the dimension $d$ is not fixed?
\end{openproblem}

As a first step towards understanding the high-dimensional setting, we show that for the general class of ``well-behaved'' functions (for which \cref{infthm:weak-approx,infthm:strong-approx} hold) finding approximate equilibrium points becomes intractable when $d$ is not fixed.

\begin{informaltheorem}
There is no polynomial-time algorithm for finding (weak) approximate equilibrium points of well-behaved functions when $d$ is not fixed, unless $\np/ = \p/$.
\end{informaltheorem}

Finally, going back to constant dimension, we consider potentials that go beyond the class of well-behaved functions (which contain the electrostatic potentials as a special case). Namely, we consider potentials that are continuously differentiable, have a singularity at the location of the charge, and decrease monotonically as we move away from the charge. Here we show that, already in two dimensions, it is intractable to find an approximate equilibrium point of a system with two charges with these very general potential functions.

\begin{informaltheorem}\label{infthm:cls-hard}
There is no polynomial-time algorithm for finding (weak) approximate equilibrium points of two charges in $\mathbb{R}^2$ with general potential functions, unless $\cls/ = \fp/$.
\end{informaltheorem}

\cls/~\cite{DaskalakisP2011-CLS,FGHS22} is a subclass of \tfnp/, the class of computational search problems having solutions that are guaranteed to exist, and once they are found, they are easily verifiable. It is believed that \cls/ is different from \fp/ (the class of search problems solvable in polynomial time), and this is supported by various cryptographic lower bounds~\cite{BitanskyPR15-Nash-crypto,ChoudhuriHKPRR19-Fiat-Shamir,JawaleKKZ21-PPAD-LWE}.

Interestingly, the best upper bound we obtain for the computational problem in \cref{infthm:cls-hard} is membership in the class \ppad/, a superclass of \cls/.

\begin{openproblem}
What is the complexity of computing an approximate equilibrium point of a general potential with singularities? Is it \cls/-complete, \ppad/-complete, or something in between?
\end{openproblem}

Given that $\cls/ = \ppadpls/$~\cite{FGHS22}, a natural first step would be to investigate whether the problem lies in \pls/, i.e., whether it admits a local search algorithm where every individual step can be performed in polynomial time (but the number of steps need not be polynomial).

\subsection{Technical Overview}

In this section, we give a brief overview of the techniques we use to obtain our results.

\subsubsection{Algorithms in constant dimension}

To keep the exposition simple in this overview, we focus on the case where all charges are positive.
Recall that our goal is to compute approximate equilibrium points of an electrostatic potential.

\paragraph{\bf Two charges.}
In the case where there are only two particles, it is easy to compute a solution to any desired accuracy. Indeed, if we consider the segment connecting the two charges, then the force pushes inwards on both ends. As a result, using binary search, one can compute a point that is $\eps$-close to an equilibrium point in time $\log(1/\eps)$.

\medskip

However, for three or more particles, it becomes unclear how to proceed.
We begin by presenting two natural approaches and why they fail (or only work in restricted cases).

\paragraph{\bf Attempt 1: Gradient descent.}
Since we are looking for a point where the gradient of the potential is zero, it is natural to attempt to use gradient descent (or ascent). These gradient methods are normally guaranteed to converge to a point with zero gradient, but in our setting this is not the case, for the following two reasons:
\begin{itemize}
\item If we use gradient ascent, then the algorithm might just start ascending towards the position of a particle, where there is no solution. Indeed, there is a singularity at the position of each charge, and the value of the potential goes to infinity as we approach it.
\item If we use gradient descent, then the algorithm might just diverge further and further away from our set of particles $P$, just like a particle that is being pushed away.
\end{itemize}
Furthermore, a solution may need to be a saddle-point as opposed to a local maximum or minimum. Indeed, for the ``natural'' power-law function in which the force between two points is proportional to $r^{-2}$, all solutions are saddle-points!\footnote{In the 3-dimensional case, this goes back to Maxwell \cite[Chapter VI Art.\ 112]{Maxwell1873}. In $d$ dimensions, if each charge contributes an amount proportional to $r^{1-d}$ to the force at $x$, the total potential obeys the Laplace equation, implying that solutions are saddle points.}
The take-away is that the convergence of gradient-based methods is highly dependent on the initialization, and it is unclear whether any theoretical guarantees for convergence can be proved.

\paragraph{\bf Attempt 2: System of polynomial equations.}
Optimization techniques can be very powerful but they are designed to work for very general and rather unstructured functions. In our setting however, the electrostatic potential is very structured; indeed, it is given by some simple algebraic expressions. It is thus natural to attempt to use this algebraic form.

The equations that an equilibrium point $\bx \in \mathbb{R}^d$ must satisfy, namely that $\nabla V(\bx) = 0$, can be written down explicitly. When the dimension $d$ is a constant, it is known that systems of polynomial equations (and inequalities) can be solved efficiently to high accuracy, see, e.g., \cite{Renegar92,GrigorevV88-polynomial-ineq}. Unfortunately, in our case, the equations are not polynomial, since they also involve roots. A straightforward remedy to this is to introduce additional variables to obtain roots using only polynomial terms; for example, adding a variable $y$ and the equations $y^2=x, y \geq 0$, ensures that $y$ is the square root of $x$. Thus, one can eliminate the root terms at the cost of introducing additional variables (and equations/inequalities). When the number of charges is constant, the number of root terms is also constant, and thus the resulting system still has a constant number of variables overall. In that case, the problem can be solved efficiently using the aforementioned polynomial solvers. However, if the number of charges is not constant, this approach fails, as the solvers no longer run in polynomial time when the number of variables is not constant.

\paragraph{\bf Our solution: Taylor approximation in a grid with variable coarseness.}
This second (partially) failed attempt suggests a somewhat different approach based on numerical approximation. Namely, if we could approximate the electrostatic potential by a polynomial, we could then use the aforementioned solvers for systems of polynomial equations to obtain an approximate solution. Unfortunately, it is not possible to approximate the potential induced by a charge by a single polynomial. Instead, we attempt to divide the domain into small grid cells, approximate the potential by a polynomial in each grid cell, and finally use the polynomial equation solver in each grid cell.

However, if we divide the domain into an evenly spaced grid and use the Taylor approximation of the potential in each grid cell, we encounter an issue. The quality of the approximation depends on how close the grid cell is to the singularity. The closer we are to the singularity, the worse the approximation error. In order to make the error small enough, the grid cells would need to be exponentially fine. Since we can only allow a polynomial number of grid cells overall, we cannot use this fine resolution over the whole domain. Instead, we use a custom grid that becomes finer as we move closer to the singularity.\footnote{To be more precise, the grid is fine in a certain coordinate, if we are close to the singularity in that coordinate.} An important point is that, since the potential goes to infinity as we approach the singularity, we can eliminate a small region around the singularity from the domain, because we can guarantee that there is no solution there. \cref{fig:grid-one} shows what this grid looks like for a single charge, and \cref{fig:grid} shows the final grid obtained by taking the grids induced by three different charges.

\begin{figure}[t]
\centering
\includegraphics[width=0.7\linewidth]{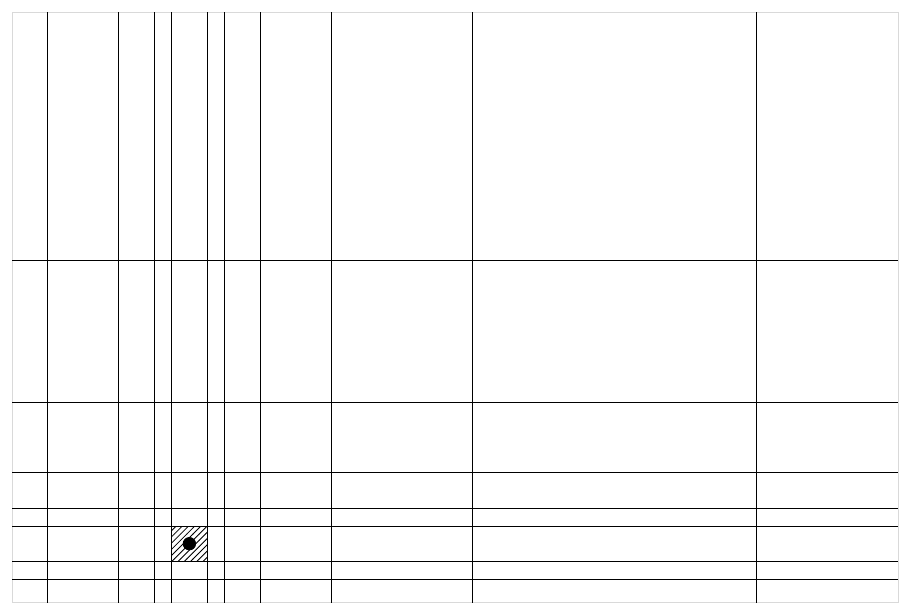}
\caption{The grid we use to ensure good approximation of the potential induced by a single charge. The charge is located at the black dot. The shaded region around the position of the charge is removed from the domain, since we can guarantee that no solution is located there. Note that the grid becomes finer in a given coordinate when we are close to the singularity in that coordinate.}
\label{fig:grid-one}
\end{figure}

\begin{figure}[t]
\centering
\includegraphics[width=0.7\linewidth]{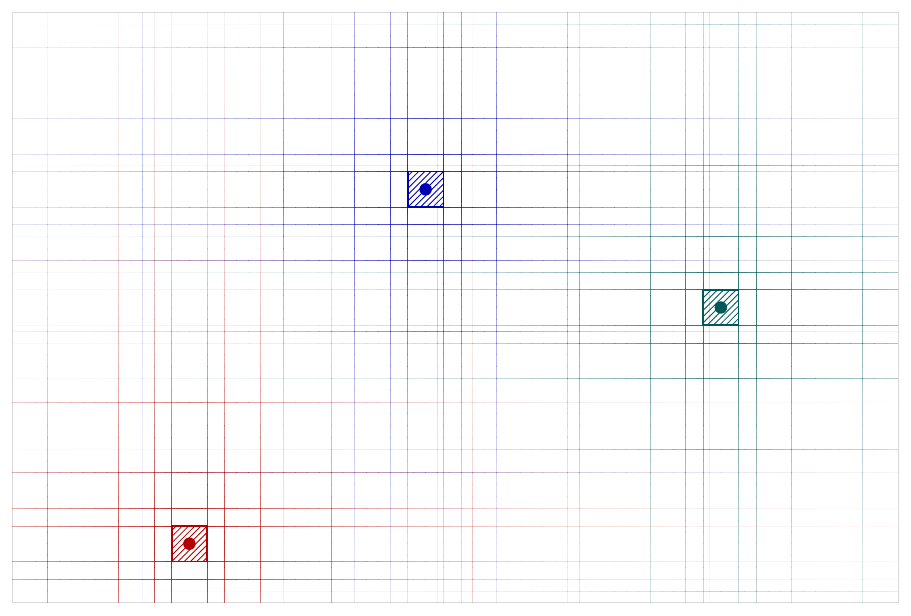}
\caption{The final grid obtained by superimposing the grids induced by three different charges. The grid lines extend to infinity, but are faded here for improved visibility. Our algorithm solves a system of polynomial inequalities in each grid cell.}
\label{fig:grid}
\end{figure}

It turns out that this approach---decomposing into a grid of variable coarseness and using Taylor approximation in each grid cell---works more generally as long as a function is \emph{well-behaved} in a sense that we define formally. Namely, we carefully define this well-behavedness property so as to capture the condition that allows for a custom grid with a polynomial number of cells to be used on the function. As a result, our algorithm can be used on a wide variety of potential functions with similar general properties.

The main limitation to our approach is that the algorithm's running time is exponential in the dimension $d$, and thus is only polynomial for constant $d$. The dependence on $d$ arises in two places: the number of grid cells, and the running time of the solvers for systems of polynomial inequalities~\cite{GrigorevV88-polynomial-ineq,Renegar92}, which we use within each grid cell.

\paragraph{\bf Strong approximation.}
As mentioned earlier, the aforementioned approach yields \emph{weak} approximate equilibrium points, i.e., points where the force is approximately zero, but not necessarily points that are close to exact equilibrium points. In order to find strong approximate solutions, we make the following observation: if a point $\bx$ is a weak approximate solution (i.e., $\nabla V(\bx) \approx 0$) and the determinant of the Hessian at $x$ is sufficiently far from zero, then there must be an exact equilibrium point close to $\bx$. The intuition is that the Hessian must have full rank in a sufficiently large region around $\bx$, and as a result, given that $\nabla V(\bx)$ is almost zero, it must take value exactly zero in some point in the vicinity of $\bx$. More formally, this can be established by using the Poincaré-Miranda theorem. Our algorithmic approach can find such a point (with gradient close to zero and determinant of the Hessian bounded away from zero) efficiently, if it exists.

\subsubsection{Generalized potentials}

As mentioned earlier, the existence of an equilibrium point of an electrostatic potential can be established using Morse theory~\cite{morse2014critical}. However, this approach has two drawbacks:
\begin{itemize}
\item It only yields existence for \emph{generic} instances, i.e., outside some set of measure zero.
\item It does not give us containment in any subclass of \tfnp/ for the associated computational problem.
\end{itemize}

Regarding the second issue, as far as we know, Morse theory has not been studied in the context of total search problems, and this is an interesting direction for future work.

\begin{openproblem}
Can one define a computational problem to capture Morse theory and use it to prove membership in some \tfnp/ class for problems where the existence is shown using this theory?
\end{openproblem}

Coming back to our electrostatic potential, we bypass both of the aforementioned issues by proposing a different proof of existence that takes care of both of these issues. Here, we give some brief intuition about how the proof can be stated for the two-dimensional setting using a generalization of the hairy ball theorem.

\paragraph{\bf A proof of existence using the hairy ball theorem.}
Let a collection $P$ of charged particles in $\bR^2$ be given and assume that the sum of the magnitudes of the charges is strictly positive.\footnote{If the sum is strictly negative, then just flip the signs of all the charges and note that this does not change the set of solutions. If the sum of magnitudes is zero, then an equilibrium point is not guaranteed to exist. Consider for example a setting with two charges, one of magnitude $+1$ and the other of magnitude $-1$.} Then, there exists a sufficiently large region $S\subseteq \bR^2$ containing $P$ such that the force points outside $S$ at any point on the boundary of $S$. Indeed, if we are far enough from the set of charges $P$, then, by our assumption, the positive charges will dominate the negative ones, and the force will point away from $P$.

Next, consider points that are very close to a positive charge $p \in P$. The force must point away from $p$. Now, remove a small region around $p$ from the set $S$. Note that the force points towards the inside of $S$ around the boundary of the small region that we removed. For each charge, we remove a sufficiently small region around it. For negative charges, the force will point outside of $S$ around the boundary of the small regions that we remove around these charges. See \cref{fig:hairy-ball} for an illustration of the set $S$.

\begin{figure}[t]
\centering
\includegraphics[width=\linewidth]{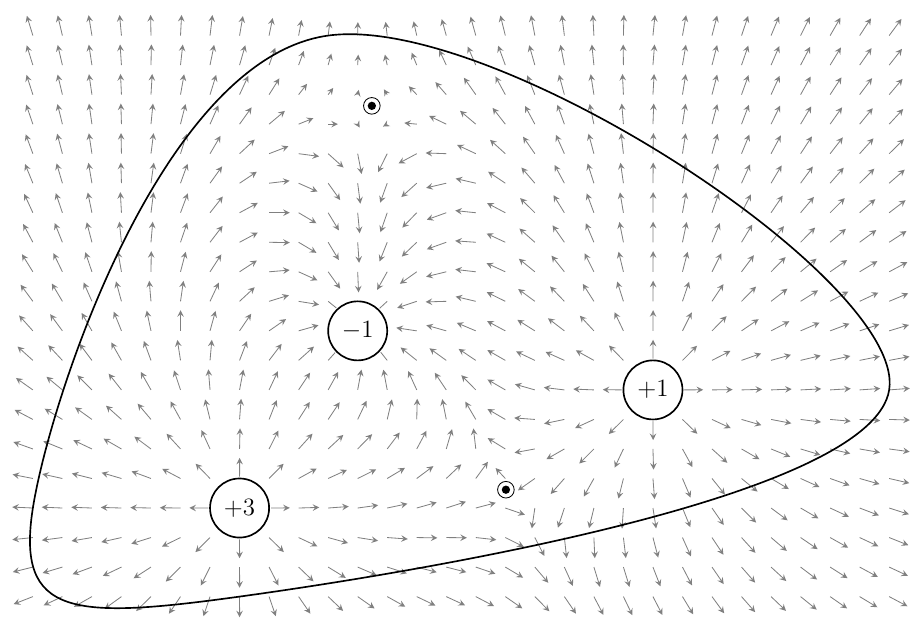}
\caption{An example with two positive charges (magnitudes $+3$ and $+1$) and one negative charge (magnitude $-1$). The force field is shown, but the length of the vectors is capped so as not to clutter the picture. The two solutions are shown as black dots. The set delimited by the black lines (inside the black outer boundary, and excluding the small round black regions) is an example of a set $S$ that we can use for our proof.}
\label{fig:hairy-ball}
\end{figure}

The last step of the proof is to take two copies of $S$, one with the force field negated, and connect them together along their corresponding boundaries, interpolating the force field smoothly at those boundaries. The resulting object lies in three-dimensional space and it is a torus of genus $|P|$, i.e., ``a donut with $|P|$ holes''. Due to the boundary conditions that the force had on the set $S$, we can argue that the ``gluing operation'' that we performed did not introduce any new solutions.\footnote{The force fields will usually not exactly agree on the boundaries that are ``glued'' together. However, they will point in the same general direction (when viewed as tangent fields) and thus it is possible to interpolate smoothly between them without introducing any new solutions.} Thus, we have obtained a continuous tangent vector field on the surface of the torus of genus $|P|$. Now, a generalization of the hairy ball theorem states that such a vector field must have a zero, and by our construction, this zero must be an equilibrium point of the original potential. Furthermore, it is known that the problem of finding a zero of a tangent vector field on such a surface lies in the class \ppad/~\cite{GoldbergH21}.

\paragraph{\bf A generalized potential problem.}
Motivated by this proof of existence, we ask the question of determining the complexity of computing an equilibrium point of a potential that satisfies certain boundary conditions. Note that our existence proof only used the fact that the potential has a set $P$ of points such that:
\begin{itemize}
\item If we move sufficiently far away from the set $P$, then the force points away from $P$.
\item For any $p \in P$, if we move sufficiently close to $p$, then the force points either away from $p$, or towards $p$.
\end{itemize}
What is the complexity of finding an equilibrium point of a potential that satisfies these conditions? The problem lies in \ppad/ by our existence proof. We provide a lower bound by showing that the problem is \cls/-hard.

In more detail, our hard instance is two-dimensional and has $|P| = 2$. The potential is obtained by taking the sum of two potentials. Each of those two potentials has a singularity at one of the two points in $P$ and decreases monotonically as we move away from that point. By a careful construction of these two potentials (which are represented as arithmetic circuits), we are able to embed an instance of a \cls/-hard problem, ensuring that any equilibrium point yields a solution to this hard problem. We reduce from the problem of finding a Karush-Kuhn-Tucker (KKT) point in a two-dimensional min-max optimization problem. This problem is known to be \cls/-hard~\cite{KalavasisPSZ25-2D-minmax}. Interestingly, the best upper bound for that problem is also \ppad/ and determining its exact complexity is a major open problem. We leave open whether a reduction can also be obtained in the other direction, and more generally the question of determining the exact complexity of the problem.

\subsection{Further Related Work}
\label{sec:related-work}

Computing equilibrium points of the electrostatic potential is arguably one of the most basic computational problems in electrostatics. 
Previous papers propose algebraic methods to solve the problem, based on general-purpose solvers whose worst-case running time is exponential in the input size~\cite{mehta2011finding,uteshev2013stationary}.
Computing an $\eps$-approximation in $\poly(1/\eps)$ time is easy in constant dimensions by creating a uniform grid, but to the best of our knowledge, there are no papers that find such a point in $\poly(\log(1/\eps))$ time. 
Our problem is also related to molecular dynamics~\cite{frenkel2023understanding}. In molecular dynamics, a large number of molecules interact with each other via intermolecular forces, including the electrostatic force, and the goal is to understand their structure and long-term dynamical behavior. 
Our problem is a much simplified version of this setting, restricted to the electrostatic force and to computing an equilibrium configuration with many fixed and one movable charge.

Historically, the emphasis has been on bounding the number of equilibrium points of the potential induced by a set of charges $P$ (in terms of the cardinality of $P$ and the dimension of space $d$) as opposed to their computation. Our results do not have new implications for the number of equilibrium points, but we hope for new insights arising from the study of their computation. Maxwell \cite{Maxwell1873} conjectured that in three dimensions ($d=3$), the number of equilibrium points is upper-bounded by $(n-1)(n-2)$; the appendix of Gabrielov et al. \cite{gabrielov2007mystery} provides a scholarly discussion of the origin of that conjecture and notes that it remains unproven. Using Milnor's theorem, Zolotov \cite{zolotov2023upper} obtained upper bounds on the number of equilibrium points for the following larger class of problems. Let $n$ be the number of point charges, each charge $i$ having some weight $q_i \in \mathbb{R}\setminus \{0\}$. The contribution to the potential $V$ from any point obeys a power law, contributing $q_i/{r_i}^m$, where $r_i$ is the distance to the point $i$. The ``real life'' power law would set $m=1$, but in principle, other values can be considered; the domain of interest is restricted to a subspace, or a superspace, of the one spanned by the $n$ fixed charges. If $m$ is even, Zolotov's upper bound is exponential in $d$, but if $m$ is odd (which includes the ``real life'' scenario), the upper bound has an exponent of $d+n$. This improves an earlier upper bound due to \cite{gabrielov2007mystery}. A line of work obtains sharper bounds for two to four dimensions ($2 \le d \le 4$)~\cite{Killian2009,lee2022nine,giorgadze2021equilibrium,uteshev2016maxwell}.

Electric fields (in high dimension) are the basis of the {\em Poisson Flow generative model}
\cite{XuLTJ22} in the context of generative models in machine learning. Data points (such as images) lie in high-dimensional space and are treated as positively charged particles in space augmented with an additional dimension, denoted $z$. Assuming data points lie on the $z=0$ hyperplane, a notional positively charged point with $z>0$ tends to be repelled in a direction of increasing $z$. As the distance $r$ from the data points increases, the distribution of directions tends to become uniform. This suggests a generative model based on: starting from a uniformly sampled point on a hemisphere with large $r$, follow in reverse the path taken. A further paper \cite{XuLTTTJ23} extends this approach to a multi-dimensional augmentation of data space: data points sit in a $\bm{z}=\bm{0}$ subspace, for a vector $\bm{z}$ of additional components. These models are incorporated into a broader class of physics-based generative models in \cite{LLXJT23}.

\section{Notation and Preliminaries}
Let $[n] = \{1, 2, \ldots, n\}$ for any positive integer $n \in \bZ_{>0}$.
Let $\indicator(\cdot)$ denote the indicator function.

\paragraph{\bf Partial Derivatives.}
For a smooth (repeatedly differentiable) function $f : X \rightarrow \bR$, where $X \subseteq \bR^d$, we denote the partial derivatives using a multi-index notation as follows:
Let $\bs \in \bZ_{\ge 0}^d$ be a $d$-dimensional multi-index, i.e., a $d$-tuple with non-negative integers. 
Let $|\bs| = ||\bs||_{1} = \sum_{j \in [d]} s_{j}$ be the sum of the elements in $\bs$
and $\bs ! = \prod_{j \in [d]} (s_{j}!)$ be the product of the factorials of the elements in $\bs$.
For an $\bx \in \bR^d$, let $\bx^s = \prod_{j \in [d]} x_j^{s_j}$.
Let $D^{\bs}$ denote the following partial derivative $D^{\bs} = \prod_{j \in [d]} \left( \frac{\partial}{\partial x_{j}} \right)^{s_{j}}$, i.e., $D^{\bs}$ partially differentiates a function with respect to the first variable $s_1$ times, the second variable $s_2$ times, and so on.
Let $M^{(k)}_f(\bx)$ denote the maximum absolute value among all $k$-th order partial derivatives of $f$ at $\bx \in X$, i.e.,
\begin{equation*}
    M^{(k)}_f(\bx) = \max_{|\bs| = k} |(D^{\bs} f) (\bx)|.
\end{equation*}
The gradient of $f$ is denoted by $\nabla f = (\frac{\partial f}{\partial x_j})_{j \in [d]}$.

\paragraph{\bf Stationary Points.}
We define weak and strong stationary points of a function $f$ in a domain $X$.
\begin{definition}[Weak Approximate Stationary Point]
A point $\bx \in X$ is a weak $\varepsilon$-approximate stationary point, or simply a weak $\varepsilon$-stationary point, for $\varepsilon > 0$ if 
\[
    || \nabla  f(\bx) ||_{\infty} = M^{(1)}_f(\bx) \le \varepsilon.
\]
\end{definition}

\begin{definition}[Strong Approximate Stationary Point]
A point $\bx \in X$ is a strong $\varepsilon$-approximate stationary point, or simply a strong $\varepsilon$-stationary point, for $\varepsilon > 0$ if there exists $\bx^* \in X$ such that
\[
    \nabla f (\bx^*) = 0 \text{ and } || \bx - \bx^* ||_{\infty} \le \varepsilon.
\]
\end{definition}

\paragraph{\bf Taylor Expansion.}
The Taylor expansion of $f$ at a point $\bx$, with a $k$-th order remainder term, with respect to a point $\bhx \in X$, assuming $\{ \bhx + t (\bx - \bhx) \mid t \in [0,1] \} \subseteq X$, is 
\begin{align}\label{eq:taylor-general}
    f(\bx) = \sum_{|\bs| < k } \frac{ (D^{\bs} f) (\bhx)  }{\bs !} (\bx - \bhx)^{\bs} + \sum_{|\bs| = k } R_{\bs}(\bx) (\bx - \bhx)^{\bs},
\end{align}
where $R_{\bs}(\bx)$ is the remainder. 
There are several equivalent formulations of $R_{\bs}(\bx)$; we are primarily interested in bounding it and will make use of the following inequality based on Lagrange's formulation
\begin{align}\label{eq:taylor-remainder}
    | R_{\bs}(\bx) | \le \frac{1}{\bs !} \max_{t \in [0,1]} \max_{|\bs'| = |\bs|} | (D^{\bs'} f)(\bhx + t (\bx - \bhx)) | = \frac{1}{\bs !} \max_{t \in [0,1]} M^{(|\bs|)}_f(\bhx + t (\bx - \bhx)) .
\end{align}

\subsection{The Electrostatic Potential (Equilibrium Points)}\label{sec:prelim:elec}
We have a configuration of $n$ charges in the $d$-dimensional Euclidean space: the $i$-th charge, $i \in [n]$, has a strength of $q_i \in \bR_{\neq 0}$ and is located at $\ba_i = (a_{i,1}, a_{i,2}, \ldots, a_{i,d}) \in \bR^d$. The electrostatic potential at a point $\bx \in \bR^d$ due to this configuration of charges is
\begin{equation}\label{eq:potential}
    f(\bx) = \sum_{i \in [n]} \frac{q_i}{|| \bx - \ba_i ||},
\end{equation}
where $|| \bx - \ba_i || = \sqrt{\sum_{j \in [d]} (x_j - a_{i,j})^2}$ is the Euclidean norm. 
A point $\bx$ is an equilibrium point (or a critical point) if $\nabla f (\bx) = 0$ and $\bx \neq \ba_i$ for all $i \in [n]$; essentially, $\bx$ is a stationary point of $f$, except that it also avoids the points of singularity $(\ba_i)_{i \in [n]}$. With a physics interpretation, an equilibrium point is where the force $- \nabla f (\bx) = 0$. The $\ell$-th component of $\nabla f (\bx)$ is
\begin{align}\label{eq:force}
    \frac{\partial f}{\partial x_{\ell}}(\bx) = \sum_{i \in [n]} \frac{ -q_i (x_{\ell} - a_{i,\ell})}{ ( \sum_{j \in [d]} (x_j - a_{i,j})^2 )^{3/2}}, \text{ for all $\ell \in [d]$},
\end{align}
where the $i$-th charge contributes $\frac{ -q_i (x_{\ell} - a_{i,\ell})}{ ( \sum_{j \in [d]} (x_j - a_{i,j})^2 )^{3/2}}$ to $\nabla f (\bx)$.
Note that the case of $d=3$ corresponds to the real-life electrostatic (Coulomb) force in the $3$-dimensional physical space.

We make the following normalizations without loss of generality: let $\min_{i \in [n]} |q_i| = 1$ and let $\min_{i, i' \in [n], i \neq i'} \max_{j \in [d]} | a_{i,j} - a_{i',j} | = 1 $. Essentially, we scale the charges to have the minimum strength of the charges to be $1$ and scale the distances to have the minimum distance among the charges for at least one coordinate to be $1$, which also implies that the minimum distance among the charges is at least $1$. Let $q_{max} = \max_i |q_i|$ and $a_{max} = \max_{i, i' \in [n], j \in [d]} | a_{i,j} - a_{i',j} |$.

\section{Computing Approximate Solutions in Constant Dimension}\label{sec:weak}

In this section, we prove that we can efficiently compute approximate stationary points of the electrostatic potential, as long as the dimension $d$ is constant.

\begin{theorem}\label{thm:electro-weak}
Fix any $d \geq 1$. Given $n$ charges in $d$-dimensional space, and a bounded convex set $X$ (given as a system of linear inequalities), as well as error parameters $\eps > \delta > 0$, we can output in polynomial time
\begin{itemize}
    \item either a point $\bx \in X$ with $\|\nabla f(\bx)\|_\infty \leq \eps$,
    \item or that there is no point $\bx \in X$ such that $\|\nabla f(\bx)\|_\infty \leq \delta$.
\end{itemize}
Here $f$ denotes the electrostatic potential created by the charges.
\end{theorem}

Our proof proceeds in two steps.
\begin{description}
\item[Step 1:] We show that, in constant dimension, we can efficiently solve systems of inequalities involving \emph{well-behaved} functions, a notion which we introduce. The main idea is to use the well-behavedness property to subdivide the domain into a polynomial number of regions, and then to apply a Taylor approximation in each region. The regions need to be chosen very carefully, and their sizes vary exponentially.
\item[Step 2:] We show that the electrostatic potential is well-behaved, as defined in step 1. At a very high level, the regions where some derivative is large can be determined efficiently and their size can be controlled.
\end{description}

\subsection{Step 1: Well-behaved Functions}

In this section, we formally define the notion of well-behaved functions that we use, and we state and prove our result about solving systems of inequalities involving such functions, in constant dimension.

We study functions of type $f : X \to \bR$ defined on compact domains $X \subseteq \bR^d$,
where $X$ is defined using a system of linear inequalities.
We assume the functions to be \textit{smooth}, i.e., $\infty$-continuously differentiable in a suitably large open set that contains $X$.
Furthermore, we assume the functions to be \textit{well-behaved}, as defined below.

For a rational number $x$, let $\sz(x)$ denote the bit-length of the representation of $x$ as an irreducible fraction, where the numerator and denominator are represented in binary. Similarly, for a rational vector $\bx \in \mathbb{R}^d$, we let $\sz(\bx)$ be the length of its representation, where each entry is represented as above. We use $a = \poly(b_1, \dots, b_n)$ to denote that there exist $k, C \in \mathbb{N}$ such that $a \leq C \cdot b_1^k \cdots b_n^k$.

\begin{definition}\label{def:efficiently-rep}
A family $\mathcal{F}$ of efficiently representable smooth functions in $d$-dimensional space is a mapping that assigns to every bitstring $r \in \{0,1\}^*$ a corresponding smooth function $f_r: X_r \to \mathbb{R}$, where $X_r \subseteq \mathbb{R}^d$ is bounded. Furthermore:
\begin{itemize}
    \item There is a polynomial-time algorithm that, given $r \in \{0,1\}^*$ as input, outputs the bounded domain $X_r \subseteq \mathbb{R}^d$ of the function $f_r$ as a system of linear inequalities.
    \item There is a polynomial-time algorithm that, given as input the representation $r \in \{0,1\}^*$, a rational point $\bx \in X_r$, a multi-index $\bs$, and $\varepsilon > 0$, outputs $v \in \mathbb{R}$ such that $|(D^{\bs} f_r)(\bx) - v| \leq \varepsilon$. To be precise, the algorithm runs in polynomial time in $|r|, \sz(\bx), |\bs|$, and $\log(1/\varepsilon)$.
\end{itemize}
\end{definition}

\begin{definition}[Well-Behaved]\label{def:well-behaved}
A family $\mathcal{F} = (f_r)_r$ of efficiently representable smooth functions in $d$-dimensional space is \emph{well-behaved} if, for every $r \in \{0,1\}^*$, there exist integers $B = \poly(|r|)$ and $C = \poly(|r|)$, such that, for every $\beta > 0$, there exist $n_{\beta} \in \bZ_{\ge 0}$ and $((L_{\beta,i,j}, H_{\beta,i,j})_{j \in [d]})_{i \in [n_{\beta}]} \in \mathbb{R}^{2 n_{\beta} d}$ such that
\begin{align*}
    n_{\beta} &= \poly(|r|, \sz(\beta)), \\
    \left\{ \bx \in X_r \ \Big| \ M^{(k)}_{f_r}(\bx) > C^k 2^{B} \frac{k!}{\beta^k} \text{ for some $k \in \bZ_{\ge 0}$} \right\} 
    &\subseteq \bigcup_{i \in [n_\beta]} \left( \bigtimes_{j \in [d]} [L_{\beta,i,j}, H_{\beta,i,j}] \right), \\
    | H_{\beta,i,j} - L_{\beta,i,j} | &\le \beta, \qquad \text{for all $i \in [n_\beta], j \in [d]$.}
\end{align*}
Additionally, we assume that $M^{(k)}_{f_r}(\bx) \le C^k 2^{B} \frac{k!}{\beta_{min}^k}$ for all $\bx \in X_r$ and $k \in \bZ_{\ge 0}$, for some value $\beta_{min} > 0$ with $\sz(\beta_{min}) = \poly(|r|)$. Furthermore, given $r$ and $\beta$, we can compute $n_{\beta}$ and $((L_{\beta,i,j}, H_{\beta,i,j})_{j \in [d]})_{i \in [n_{\beta}]}$ in time polynomial in $|r|$ and $\sz(\beta)$.
\end{definition}

Our main result in this section is the following.

\begin{theorem}\label{thm:well-behaved-easy}
Fix any $d \geq 1$. For any well-behaved family $\mathcal{F}$ of efficiently representable smooth functions in $d$-dimensional space, we can solve the following problem in polynomial time. Given functions $f_1, \dots, f_m \in \mathcal{F}$ (given by their representations $r_1, \dots, r_m$), an error parameter $\eps > 0$, and a bounded domain $X \subset \mathbb{R}^d$ (given as a system of linear inequalities) that satisfies $X \subseteq \cap_{i=1}^m X_{r_i}$, output:
\begin{itemize}
    \item either a point $\bx \in X$ that satisfies $f_\ell(\bx) \leq 0$ for all $\ell \in [m]$,
    \item or that there is no point $\bx \in X$ that satisfies $f_\ell(\bx) \leq -\eps$ for all $\ell \in [m]$.
\end{itemize}
\end{theorem}

Before providing the proof of this theorem, we give a brief proof sketch.
The starting point of our algorithm is to use Taylor expansions to approximate the functions $f_1, \dots, f_m$. Assuming that we could do this, we could then use existing tools to solve systems of polynomial inequalities in constant dimension~\cite{GrigorevV88-polynomial-ineq,Renegar92}.

The challenge comes from the fact that the Taylor expansion is only a good approximation of the function close to the point around which it is created. The notion of well-behavedness that we define can be used to ensure that the domain can be subdivided into a polynomial number of regions, such that in each of them there is a good Taylor approximation. These regions are not a nice regular grid, but instead 
the grid is finer where the functions $f_\ell$ change rapidly (derivatives are large) and coarser where the functions $f_\ell$ change slowly.

\begin{proof}[Proof of \cref{thm:well-behaved-easy}]
Fix the dimension $d \geq 1$ and let $\mathcal{F} = (f_r)_r$ be any well-behaved family of efficiently representable smooth functions in $d$-dimensional space. Recall that we are given functions $f_1, \dots, f_m \in \mathcal{F}$ (given by their representations $r_1, \dots, r_m$), an error parameter $\eps > 0$, and a bounded domain $X \subset \mathbb{R}^d$ (given as a system of linear inequalities) that satisfies $X \subseteq \cap_{i=1}^m X_{r_i}$.

Our goal is to provide an algorithm that runs in time polynomial in $m$, $|r_1|, \dots, |r_m|$, $\log(1/\epsilon)$, and the size of the representation of $X$ and which outputs:
\begin{itemize}
    \item either a point $\bx \in X$ that satisfies $f_\ell(\bx) \leq 0$ for all $\ell \in [m]$,
    \item or that there is no point $\bx \in X$ that satisfies $f_\ell(\bx) \leq -\eps$ for all $\ell \in [m]$.
\end{itemize}
Note that the algorithm can take super-polynomial time in $d$, as we are assuming $d$ to be a constant.

Our algorithm uses a piecewise Taylor expansion. 
We divide the search space into an axis-aligned grid. 
For each cell in the grid, we construct $m$ polynomials that approximate the $m$ functions $(f_\ell)_{\ell \in [m]}$.
We then solve a system of polynomial inequalities for each cell in the grid to find an approximate solution, if any.
The grid is finer where the functions $f_\ell$ change rapidly (derivatives are large) and coarser where the functions $f_\ell$ change slowly. This ensures that our polynomial approximation has a small error.
We use the well-behavedness property of the family $\mathcal{F}$ to give a polynomial bound on the number of required grid cells.

\paragraph{\bf Grid.}
We construct the grid using a set of axis-aligned hyperplanes.
Fix an $\ell \in [m]$ and focus on the function $f_\ell$; we repeat the same process for each $\ell \in [m]$.
Note that $f_\ell$ comes from the family $\mathcal{F}$ of well-behaved functions. We can thus compute the parameters $B_\ell$, $C_\ell$, and $\beta_{\ell,min}$ associated with $f_\ell$ as introduced in \cref{def:well-behaved}. Furthermore, we also compute, for each $j \in [d]$
$$L_j = \min_{\bx \in X} x_j, \quad H_j = \max_{\bx \in X} x_j$$
and $t_{\ell,max}$, which is defined as the smallest $t \geq 1$ such that
$$\beta_{\ell,min} 2^{t} \geq \max_{j \in [d]} | H_j - L_j|.$$
Finally, we let $\beta_{\ell,max} = \beta_{\ell,min} 2^{t_{\ell,max}}$.
Since $X$ is bounded and given as a system of linear inequalities, we can compute these quantities in polynomial time. Furthermore, $t_{\ell,max}$ is an integer bounded by a polynomial.

For $t = 0$ to $t_{\ell,max}$:
\begin{itemize}
    \item Let $\beta = \beta_{\ell,min} 2^t$.
    \item Compute the $n_{\ell,\beta}$ hypercuboids $((L_{\ell,\beta,i,j}, H_{\ell,\beta,i,j})_{j \in [d]})_{i \in [n_{\ell,\beta}]}$ satisfying the conditions set out in \cref{def:well-behaved}. Note that by the well-behavedness of $\mathcal{F}$, these are guaranteed to exist and can be computed efficiently.
    \item For every such hypercuboid $(L_{\ell,\beta,i,j}, H_{\ell,\beta,i,j})_{j \in [d]}$, where $i \in [n_{\ell,\beta}]$, add the following hyperplanes:
    \begin{equation*}
        \left\{ \bx \in \bR^d \mid x_j = L_{\ell,\beta,i,j} + s \frac{H_{\ell,\beta,i,j} - L_{\ell,\beta,i,j}}{4 C_\ell d}  \right\}, \text{ for $j \in [d]$ and } s \in \{0\} \cup [ 4 C_\ell d ].
    \end{equation*}
    In other words, we split each hypercuboid into $4 C_\ell d$ uniform pieces along every axis.
\end{itemize}
We also add the hyperplanes $\{ \bx \in \bR^d \mid x_j = L_j \}$ and $\{ \bx \in \bR^d \mid x_j = H_j \}$ for all $j \in [d]$, and also split each of those hypercuboids into $4 C_\ell d$ uniform pieces along every axis.

Let us count the total number of axis-aligned hyperplanes we added along each axis $j \in [d]$. We added two hyperplanes $x_j = L_j$ and $x_j = H_j$ at the ends, as well as $4 C_\ell d - 1$ additional hyperplanes evenly spaced between them. For each $\ell \in [m]$, we iterated over $t_{\ell,max} + 1$ values of $\beta$; for each $\beta$, we iterated over $n_{\ell,\beta}$ hypercuboids; for each hypercuboid, we added $4 C_\ell d + 1$ hyperplanes. Overall, we have at most
\[
    \sum_{\ell \in [m]} (1 + 4 C_\ell d ) \left( 1 + \sum_{t = 0}^{t_{\ell,max}} n_{\ell, \beta_{\ell,min} 2^{t}}\right) ,
\]
hyperplanes along each axis, which is polynomial in the input size.

\paragraph{\bf Polynomial Approximation.}
The set of hyperplanes specified above divides the $\bigtimes_{j \in [d]} [L_j, H_j]$ hypercuboid that covers $X$ into a grid. For each cell in the grid, where a cell is a minimal hypercuboid created by these hyperplanes that does not contain another hypercuboid, we can find an approximate solution, if it exists, by Taylor-approximating each function $f_\ell$ for $\ell \in [m]$. 

Fix a cell in the grid---denoted by $C = \bigtimes_{j \in [d]} [L_j', H_j']$ and let $X' = C \cap X$. If $X'$ is empty, then we ignore this cell. Otherwise, pick an arbitrary point $\bhx \in X'$. We approximate $f_\ell$ using a polynomial of degree $k_\ell-1$, where $k_\ell = \lceil B_\ell + \log(8/\epsilon) \rceil$.
In particular, we compute the $(k_\ell-1)$-th order Taylor expansion of $f_\ell$ with respect to the point $\bhx$, denoted by $g_\ell(\bx)$
\begin{align}
    g_\ell(\bx) = \sum_{|\bs| < k_\ell } \frac{ (D^{\bs} f_\ell) (\bhx)  }{\bs !} (\bx - \bhx)^{\bs}.
\end{align}
Using the fact that the family of functions is well-behaved, we obtain the following crucial claim, which we prove at the end of this section.

\begin{claim}\label{clm:taylor-approx}
The $(k_\ell-1)$-th Taylor approximation of $f_\ell$ with respect to some arbitrary point $\bhx \in X'$, denoted by $g_\ell$, has an error of at most $\epsilon/8$, i.e.,
\[
    \left| f_\ell(\bx) - g_\ell(\bx) \right| \le \epsilon/8, \text{ for all $\ell \in [m]$ and $\bx \in X'$.}
\]
\end{claim}

In general, we might not be able to compute the values $(D^{\bs} f_\ell) (\bhx)$ exactly. However, by \cref{def:efficiently-rep}, we know that we can approximate these values arbitrarily well. Thus, by computing the coefficients with sufficient precision, we obtain a polynomial $\widehat{g}_\ell$ of degree $k_\ell-1$ that satisfies $|\widehat{g}_\ell(\bx) - g_\ell(\bx)| \leq \eps/8$ for all $\bx \in X'$. As a result, together with the claim, we have that
\begin{equation}\label{eq:taylor-error}
\left| f_\ell(\bx) - \widehat{g}_\ell(\bx) \right| \le \epsilon/4, \text{ for all $\ell \in [m]$ and $\bx \in X'$.}
\end{equation}

Next, we will make use of the following theorem \cite{GrigorevV88-polynomial-ineq} (an algorithm with improved complexity is given in \cite{Renegar92}). Note that here it is crucial that the dimension $d$ is constant.

\begin{theorem}[e.g.~\cite{GrigorevV88-polynomial-ineq}]
Fix $d \geq 1$. Given $\delta > 0$ and a system of polynomial inequalities in $d$ variables, $p_i(\bx) \leq 0$ for all $i \in [m]$, we can output in polynomial time
\begin{itemize}
    \item either a point $\bx \in \mathbb{R}^d$ such that there exists $\bx^* \in \mathbb{R}^d$ with $\|\bx-\bx^*\| \leq \delta$ and $p_i(\bx^*) \leq 0$ for all $i \in [m]$,
    \item or that there is no point $\bx \in \mathbb{R}^d$ that satisfies $p_i(\bx) \leq 0$ for all $i \in [m]$.
\end{itemize}
\end{theorem}

We apply this theorem to the following system of polynomial inequalities
\begin{align*}
    \widehat{g}_\ell(\bx) + \eps/2 &\le 0, &\text{for $\ell \in [m]$},\\
    \bx &\in X', &
\end{align*}
where the last constraint represents the list of linear inequalities defining the feasible region $X'$. We set $\delta = \eps/4L$, where $L$ is such that all the $\widehat{g}_\ell$ are $L$-Lipschitz-continuous over some sufficiently big compact set containing $X'$.

If the solver outputs that the system is not satisfiable, then, since $|\widehat{g}_\ell(\bx) - f_\ell(\bx)| \leq \eps/4$ by \eqref{eq:taylor-error}, it follows that there is no $\bx \in X'$ such that $f_\ell(\bx) \leq - \eps$ for all $\ell \in [m]$.

If the solver outputs that the system is satisfiable, then it also outputs a point $\bx \in \mathbb{R}^d$ such that there exists $\bx^* \in X'$ with $\|\bx-\bx^*\| \leq \delta$ and $\widehat{g}_\ell(\bx^*) + \eps/2 \le 0$ for all $\ell \in [m]$. By replacing $x$ by its projection\footnote{The projection can be computed in polynomial time, since $X'$ is given as a system of linear inequalities.} onto the convex set $X'$, we have that $x \in X'$ and $\|\bx-\bx^*\| \leq \delta$. As a result, it follows that
$$\widehat{g}_\ell(\bx) \leq \widehat{g}_\ell(\bx^*) + L \|\bx-\bx^*\| \leq -\eps/2 + L\delta \leq -\eps/2 + \eps/4 \leq -\eps/4.$$
Since $|\widehat{g}_\ell(\bx) - f_\ell(\bx)| \leq \eps/4$ by \eqref{eq:taylor-error}, we thus have $f_\ell(\bx) \leq 0$ for all $\ell \in [m]$.

After running the solver on each cell, we either have found a point $\bx \in X$ with $f_\ell(\bx) \leq 0$ for all $\ell \in [m]$, or, we can confidently output that there exists no $\bx \in X$ such that $f_\ell(\bx) \leq - \eps$ for all $\ell \in [m]$.

\paragraph{\bf Running Time.}
Since the number of axis-aligned hyperplanes along each axis is polynomial, and since the dimension $d$ is constant, the number of cells is polynomial. For each such cell we run the polynomial inequality solver, which runs in polynomial time, because the dimension $d$ is constant. Thus, the total running time of the algorithm is polynomial.
\end{proof}

\begin{proof}[Proof of \cref{clm:taylor-approx}]
Let us fix an $\ell \in [m]$ and focus on $f_\ell$; the same argument applies for every $\ell$. We suppress the notation $\ell$ to decrease clutter.

From \eqref{eq:taylor-general} and \eqref{eq:taylor-remainder}, we know that for any $\bx \in X'$ the error is bounded by
\begin{align*}
    \left| f(\bx) - g(\bx) \right| 
    &\le \left| \sum_{|\bs| = k } R_{\bs}(\bx) (\bx - \bhx)^{\bs} \right| 
    \le \sum_{|\bs| = k } | R_{\bs}(\bx) | |\bx - \bhx|^{\bs} \\
    &\le \sum_{|\bs| = k } \frac{1}{\bs !} \left( \max_{t \in [0,1]} M^{(|\bs|)}_f(\bhx + t (\bx - \bhx)) \right) |\bx - \bhx|^{\bs} \\
    &\le \max_{t \in [0,1]} M^{(k)}_f(\bhx + t (\bx - \bhx))  \sum_{|\bs| = k } \frac{1}{\bs !} |\bx - \bhx|^{\bs}.
\end{align*}
Let $\alpha = \max_{\by \in X' } M^{(k)}_f(\by)$. 
Since the family of functions is well-behaved (\cref{def:well-behaved}), we have that $M^{(k)}_f(\bx) \le C^k 2^{B} \frac{k!}{\beta_{min}^k}$ for all $\bx \in X$. So it must be that $\alpha \le C^k 2^{B} \frac{k!}{\beta_{min}^k}$. 
In the algorithm, $\beta$ is incremented by a factor of $2$, starting from $\beta_{min}$ and going to $\beta_{max}$.
Let $\beta'$ be the smallest value among the set of $\beta$-values used in the algorithm such that $\alpha > C^k 2^{B} \frac{k!}{\beta'^k}$. If there is no such value, then it must be that $\alpha \le C^k 2^{B} \frac{k!}{\beta_{max}^k}$, and we set $\beta' = \beta_{max}$.
Notice that by the definition of $\beta'$, we have 
\begin{equation}\label{eq:lm:error:1}
    \alpha = \max_{\by \in X'} M^{(k)}_f(\by) \le C^k 2^{B} \frac{k!}{(\beta'/2)^k}.
\end{equation}
In the case where $\alpha > C^k 2^{B} \frac{k!}{\beta_{max}^k}$, we have $\alpha > C^k 2^{B} \frac{k!}{\beta'^k}$. But this means that a point inside the cell $X'$ (namely the one where the value $\alpha$ is achieved) must be covered by a hypercuboid corresponding to $\beta'$. Since $X'$ is a cell, all of $X'$ must thus be covered by that hypercuboid, which has sides of length at most $\beta'$. Further, in the algorithm, this hypercuboid was cut into smaller pieces by splitting each side into $4Cd$ parts. Therefore, we must have 
\begin{equation}\label{eq:lm:error:2}
\max_{j \in [d], \bx \in X'} |x_j - \hx_j| \le \frac{\beta'}{4 C d} .
\end{equation}
In the case where $\alpha > C^k 2^{B} \frac{k!}{\beta_{max}^k}$, we have $\beta' = \beta_{max}$, and \eqref{eq:lm:error:2} again holds, since the whole domain $X$ fits within a box of side-length $\beta_{max}$, and we further subdivided that box into $4Cd$ parts in a uniform manner.

Furthermore, by the multinomial theorem
\begin{equation}\label{eq:lm:error:3}
    d^k = (1 + 1 + \ldots + 1)^k =  \sum_{|\bs| = k } \frac{k!}{\bs !} \Longleftrightarrow \sum_{|\bs| = k } \frac{1}{\bs !} = \frac{d^k}{k!}.
\end{equation}
Plugging \eqref{eq:lm:error:1}, \eqref{eq:lm:error:2}, and \eqref{eq:lm:error:3} into the error bound, we get
\begin{align*}
    \left| f(\bx) - g(\bx) \right| &\le \max_{t \in [0,1]} M^{(k)}_f(\bhx + t (\bx - \bhx))  \sum_{|\bs| = k } \frac{1}{\bs !} |\bx - \bhx|^{\bs} \\
    &\le \max_{ \by \in X'} M^{(k)}_f(\by)  \sum_{|\bs| = k } \frac{1}{\bs !} \left( \frac{\beta'}{4 C d} \right)^k \\
    &\le C^k 2^{B} \frac{k!}{(\beta'/2)^k} \frac{d^k}{k!} \left( \frac{\beta'}{4 C d} \right)^k = \frac{2^B}{2^k} \le \epsilon/8,
\end{align*}
where the last inequality holds because $k = \lceil B + \log(8/\epsilon) \rceil \ge B + \log(8/\epsilon)$.
\end{proof}

\subsection{Step 2: Application to Electrostatic Potentials}

In this section, we apply the tools introduced in step 1 to the setting of electrostatic potentials. For this, we need to show that the potential $f$ is well-behaved, as defined in the previous section.
However, notice that the function $f$ and its derivatives are unbounded near the charges; in particular
\[
    \left| \frac{\partial f}{\partial x_{\ell}}(\bx) \right| = \left| \sum_{i \in [n]} \frac{ -q_i (x_{\ell} - a_{i,\ell})}{ ( \sum_{j \in [d]} (x_j - a_{i,j})^2 )^{3/2}} \right| \to \infty \text{ as $\bx \to \ba_i$ for any $i \in [n]$}.
\]
This means that the function is not well-behaved in its entire domain. However, we observe that a small region near each charge, where the derivatives become unbounded, never has an equilibrium point.

\begin{lemma}
\label{lm:box-around-charges}
For every $i \in [n]$, there is no $\eps$-equilibrium point in the hypercube $ \bigtimes_{j \in [d]} [a_{i,j} - \rho, a_{i,j} + \rho] $ around the charge $q_i$ located at $\ba_i$, where $\rho = \frac{1}{d\sqrt{d}(4nq_{max} + \eps)}$.
\end{lemma}
\begin{proof}
Let $\cB_i = \bigtimes_{j \in [d]} [a_{i,j} - \rho, a_{i,j} + \rho] $ be the hypercube around charge $i$.
For any $i \in [n]$, we show that for any point in $\cB_i$, the force due to charge $i$ dominates the total force by other charges.

Let us fix a charge $i \in [n]$. Let us pick a point $\bx \in \cB_i \setminus \{\ba_i\}$. Recall that by our normalization assumption we have $\min_{i, i' \in [n], i \neq i'} \max_{j \in [d]} | a_{i,j} - a_{i',j} | = 1 $. 
Thus, for any charge $i' \neq i$, there is at least one coordinate $j \in [d]$ such that $| x_j - a_{i',j} | \ge 1 - \rho$, so $\sum_{j \in [d]} (x_j - a_{i',j})^2 \ge (1 - \rho)^2$.
For each $\ell \in [d]$, the magnitude of the total contribution to $\frac{\partial f}{\partial x_{\ell}}(\bx)$ by all charges except $i$ is
\begin{align*}
    \left| \sum_{i' \neq i} \frac{ q_{i'} (x_{\ell} - a_{i',\ell})}{ ( \sum_{j \in [d]} (x_j - a_{i',j})^2 )^{3/2}} \right| 
    \le \sum_{i' \neq i} \frac{ | q_{i'} | | (x_{\ell} - a_{i',\ell}) | }{ ( \sum_{j \in [d]} (x_j - a_{i',j})^2 )^{3/2}} 
    &\le \sum_{i' \neq i} \frac{ q_{max} }{ \sum_{j \in [d]} (x_j - a_{i',j})^2 } \\
    &\le  \frac{(n-1) q_{max}}{(1 - \rho)^2} \leq 4(n-1)q_{max}
\end{align*}
where we used the fact that $\rho \leq 1/2$.
On the other hand, there exists an $\ell \in [d]$ for which the magnitude of the force due to charge $i$ along axis $\ell$, i.e., the contribution to $F_\ell$, is
\begin{align*}
    \max_{\ell \in [d]} \frac{ | q_i | |x_{\ell} - a_{i,\ell}| }{ ( \sum_{j \in [d]} (x_j - a_{i,j})^2 )^{3/2}} 
    &= \frac{|q_i|}{\sum_{j \in [d]} (x_j - a_{i,j})^2 } \max_{\ell \in [d]} \frac{ |x_{\ell} - a_{i,\ell}| }{ ( \sum_{j \in [d]} (x_j - a_{i,j})^2 )^{1/2}} \\
    &\ge \frac{|q_i|}{\sum_{j \in [d]} (x_j - a_{i,j})^2 } \frac{1}{\sqrt{d}} 
    \ge \frac{1}{\rho^2 d\sqrt{d}}
\end{align*}
where we used the fact that $\min_{i \in [n]} |q_i| = 1$, by our normalization assumption.

So, if we pick $\rho$ such that the force due to $i$ dominates the force due to other charges by more than $\eps$, then we can ensure that there are no $\eps$-equilibrium points in $\cB_i$. In particular, we can set $\rho$ such that
\begin{align}
    \frac{1}{\rho^2 d\sqrt{d}} > 4(n-1) q_{max} + \eps \impliedby \rho^2 < \frac{1}{d\sqrt{d}(4nq_{max} + \eps)}
\end{align}
and we have chosen $\rho$ small enough to satisfy this.
\end{proof}

Given \cref{lm:box-around-charges}, we can safely avoid searching for an equilibrium point inside the hypercube centered around each charge of side length $2\rho$. The region left after cutting out these hypercubes from the domain $X$ is not a convex set, which means that we cannot immediately apply our Taylor approximation algorithm (\cref{thm:well-behaved-easy}). We can easily resolve this by splitting this non-convex set into a polynomial number of convex sets by creating a grid as follows:
Along each axis $j \in [d]$ and around each charge $i \in [n]$, add the hyperplanes $a_{i,j} \pm \rho$. Overall, along each axis, we have introduced at most $2n$ hyperplanes. These $2n$ hyperplanes create $2n + 1$ segments. So, we have at most $(2n + 1)^d$ cells in the grid, which are polynomially many because $d$ is assumed to be a constant. Our convex sets are simply the intersections of these cells with $X$, where we ignore the cells within $\rho$ distance from each charge. Since $X$ is bounded, each of those cells intersected with $X$ will also be bounded.

Given this, it now remains to prove that the electrostatic potential is well-behaved in each such bounded set that does not contain the regions close to the charges. It suffices to prove this for the potential induced by one charge (since we can use \cref{lm:linear-comb} to extend it to multiple charges). Namely, we need the following lemma, which is proved at the end of this section.

\begin{lemma}\label{lm:one-charge-well-behaved}
Fix $d \geq 1$. The family of electrostatic potentials induced by a single charge
$$\mathcal{F} = \left\{f: X' \to \mathbb{R}, \bx \mapsto \frac{q}{|| \bx - \ba ||_2}: q \in \mathbb{Q}_{\neq 0}, \ba \in \mathbb{Q}^d, X' \subset \bR^d \text{ bounded with } \ba \notin X'\right\}$$
is well-behaved.\footnote{As before, we assume that $X'$ is a bounded polytope given as a system of linear inequalities.}
\end{lemma}

As mentioned above, after eliminating the regions close to the singularities, where solutions cannot occur, we obtain a polynomial number of convex domains. Let $X'$ denote any one of them.
Given error parameters $\eps > \delta > 0$, we write down the system of equations
\begin{align*}
    \frac{\partial f}{\partial x_j}(\bx) &\le \eps, &\text{for $j \in [d]$},\\
    -\frac{\partial f}{\partial x_j}(\bx) &\le \eps, &\text{for $j \in [d]$},\\
    \bx &\in X', &
\end{align*}
As $f$ is well-behaved, using the tools in \cref{app:well-behaved-properties} (derivatives of well-behaved functions), we know that $\frac{\partial f}{\partial x_j}$ is also well-behaved for all $j \in [d]$. So, by using the algorithm presented in the previous section (\cref{thm:well-behaved-easy}) with error parameter $\eps - \delta$, we obtain in polynomial time
\begin{itemize}
    \item either a point $\bx \in X'$ such that $\|\nabla f(\bx)\| \leq \eps$,
    \item or the guarantee that there is no point $\bx \in X'$ with $\|\nabla f(\bx)\| \leq \delta$.
\end{itemize}
By performing this for each of the polynomially many sets $X'$, we obtain the desired output. Namely, either we find an $\eps$-stationary point in $X$, or we can guarantee that there is no $\delta$-stationary point in $X$.

\begin{proof}[Proof of \cref{lm:one-charge-well-behaved}]
Let us first compute $M^{(k)}_f(\bx)$. For an arbitrary multi-index $\bs \in \bZ_{\ge 0}^d$ and a non-negative integer $t \in \bZ_{\ge 0}$, say we have a function $g(\bx)$ of the form
\begin{align}\label{eq:general-term}
    g(\bx) = \frac{\kappa_g (\bx - \ba)^{\bs}}{( \sum_{j \in [d]} (x_j - a_j)^2 )^{t/2}} = \frac{ \kappa_g \prod_{j \in [d]} (x_{j} - a_{j})^{s_j} }{( \sum_{j \in [d]} (x_j - a_{j})^2 )^{t/2}},
\end{align}
for some constant $\kappa_g$. Notice that the degree of the numerator of $g$ is equal to $|\bs| = \sum_j s_j$. Similarly, let us define the \textit{effective} degree of the denominator as $t$ (although the denominator is not a polynomial, it behaves like a degree-$t$ polynomial for our purposes). Let us call the ratio of the degree of the denominator vs.\ the numerator the \textit{degree-ratio} and denote it as the pair $(t, |\bs|)$. Also note that the difference between the two degrees (denominator minus the numerator) is $t - |\bs|$, let us call this the \textit{degree-difference}.
Let us now partially differentiate $g$ w.r.t.\ $x_{j'}$ for some $j' \in [d]$. We obtain
\begin{align*}
    \frac{\partial g}{\partial x_{j'}}(\bx)  = \frac{ \kappa_g s_{j'} (x_{j'} - a_{j'})^{s_{j'} - 1} \prod_{j \neq j'} (x_{j} - a_{j})^{s_j} }{( \sum_{j \in [d]} (x_j - a_{j})^2 )^{t/2}} - \frac{ \kappa_g t (x_{j'} - a_{j'})^{s_{j'} + 1} \prod_{j \neq j'} (x_{j} - a_{j})^{s_j} }{( \sum_{j \in [d]} (x_j - a_{j})^2 )^{(t+2)/2}}.
\end{align*}
We can make the following observations about $\frac{\partial g}{\partial x_{j'}}(\bx)$, which we shall use later:
\begin{itemize}
    \item It has at most the two terms given above. The first term appears if $s_{j'} > 0$. The second term appears if $t > 0$.
    \item The constant coefficient of the first term is $\kappa_g s_{j'}$, so it is $s_{j'}$ times the constant in $g$. On the other hand, the constant coefficient of the second term is $-t$ times that of $g$. Overall, the magnitude of the coefficient of each of the terms increased by a factor of at most $\max(s_{j'}, t) \le \max(|\bs|,t)$.
    \item For the first term, the degree of the numerator is $|\bs|-1$, while the effective degree, or simply degree, of the denominator is $t$, i.e., the degree-ratio is $(t, |\bs|-1)$. So, the degree-difference is $t - (|\bs| - 1) = t - |\bs| + 1$, which is one higher than that of $g$. On the other hand, for the second term, the degree-ratio is $(t + 2, |\bs| + 1)$ and the degree-difference is $(t + 2) - (|\bs| + 1) = t - |\bs| + 1$, which is again one higher than that of $g$.
    \item The degree of the denominator of the terms either stayed the same or went up by $2$, which also implies that its parity remained the same.
\end{itemize}
Overall, each of the two terms of the partial derivative of $g$ has a form similar to $g$, albeit with slightly different exponents to the terms in the numerator and the denominator and with different coefficients. The above observations hold for any $j' \in [d]$.

The electrostatic potential due to the charge is $f(\bx) = \frac{q}{ (\sum_{j \in [d]} (x_j - a_{j})^2)^{1/2} }$. Notice that $f$ is a function of the type given in \eqref{eq:general-term} and satisfies the properties discussed above. In particular, $f$ has a degree-ratio of $(1, 0)$, a degree-difference of $1$, a coefficient of $q$, and the parity of the degree of the denominator is odd.

If we differentiate $f$ once w.r.t.\ $x_{j'}$ for any $j' \in [d]$, we will get terms with degree-difference exactly equal to $2$. Moreover, as the degree of the denominator of $f$ is $1$ and has an odd parity, all terms in $\partial f / \partial x_{j'}$ will have the degree of the denominator at most $3$ and parity odd. So, the only feasible degree-ratio is $(3, 1)$.
Similarly, let us now look at the $k$-th derivative of $f$, i.e., let us look at $D^{\bs} f$ for some multi-index $\bs$ satisfying $|\bs| = k$. Applying our observations above, $D^{\bs} f$ will have a degree difference of $k+1$. The maximum degree of the denominator will be $2k+1$. So, the possible degree ratios are, 
if $k$ is even: $(k+1, 0), (k+3, 2), \ldots, (2k+1, k)$;
if $k$ is odd: $(k+2, 1), (k+4, 3), \ldots, (2k+1, k)$.

Let us first focus on the non-constant part of the terms in $D^{\bs} f$. Based on the possible degree ratios mentioned above, each term (excluding the constant coefficient part) can be written as
\begin{align*}
    \frac{(\bx - \ba)^{\bs'}}{( \sum_{j \in [d]} (x_j - a_{j})^2 )^{t'/2}} = \frac{ \prod_{j \in [d]} (x_{j} - a_{j})^{s_j'} }{( \sum_{j \in [d]} (x_j - a_{j})^2 )^{t'/2}},
\end{align*}
evaluated at some $\bx \in X$, where $\bs' \in \bZ_{\ge 0}$ and $t' \in \bZ_{\ge 0}$ satisfying $t' - | \bs' | = k + 1$, $t' \le 2k + 1$, and $t'$ is odd. As $|x_{j'} - a_{j'}| \le \sqrt{\sum_{j \in [d]} (x_j - a_{j})^2}$ for all $j' \in [d]$, the magnitude of the above quantity can be bounded as
\begin{align}\label{eq:lm:one-charge-well-behaved:1}
    &\left| \frac{(\bx - \ba)^{\bs'}}{( \sum_{j \in [d]} (x_j - a_{j})^2 )^{t'/2}} \right| \le \frac{1}{( \sum_{j \in [d]} (x_j - a_{j})^2 )^{(t'- |\bs'|)/2}} = \frac{1}{( \sum_{j \in [d]} (x_j - a_{j})^2 )^{(k+1)/2}}.
\end{align}

Let us now focus on the constant coefficients of the terms in $D^{\bs} f$. In all the terms we differentiate, starting from $f$ to reach $D^{\bs} f$, the degree-difference is positive, i.e., the degree of the numerator is less than the degree of the denominator. 
Further, as we had noted earlier, when we take a derivative of any of these terms, each generated term has an additional multiplicative coefficient of magnitude at most the degree of the numerator or the denominator, which is upper bounded by the degree of the denominator for our problem.
Also, as observed earlier, each term can create at most two terms when differentiated by $x_{j'}$ for some $j' \in [d]$. So, each term can contribute to at most $2d$ terms after one round of differentiation. So, the total number of terms \textit{without} combining the terms with exactly the same form in the numerator (same $\bs'$) and the denominator (same $t'$) is at most $(2d)^k$. And each of these terms will have a coefficient of at most
\[
    (2k-1) \cdot (2k-3) \cdot \ldots \cdot 5 \cdot 3,
\]
where $(2k-1)$, $(2k-3)$, etc., are the maximum multipliers to the coefficient when we take the $k$-th, $(k-1)$-th, etc., derivatives of the denominator, respectively. We can simplify this expression as
\begin{align}
    (2k-1) \cdot (2k-3)  \cdot &\ldots \cdot 5 \cdot 3 
    = \frac{(2k-1) \cdot (2k-2) \cdot (2k-3) \cdot (2k-4) \ldots \cdot 5 \cdot 4 \cdot 3 \cdot 2}{(2k-2) \cdot (2k-4) \cdot \ldots \cdot 4 \cdot 2} \nonumber \\
    &= \frac{(2k-1)!}{2^{k-1} (k-1)!} = \frac{1}{2^{k}}\frac{2k}{k}\frac{(2k-1)!}{(k-1)!}\frac{k!}{k!} = \frac{k!}{2^{k}} \binom{2k}{k} \le \frac{k!}{2^{k}} \frac{4^k}{\sqrt{\pi k}} \le k! 2^k. \label{eq:lm:one-charge-well-behaved:2}
\end{align}
Putting \eqref{eq:lm:one-charge-well-behaved:1} and \eqref{eq:lm:one-charge-well-behaved:2} together, we have
\begin{equation}\label{eq:lm:one-charge-well-behaved:3}
    M^{(k)}_f(\bx) \le \frac{k! 2^k q}{( \sum_{j \in [d]} (x_j - a_{j})^2 )^{(k+1)/2}} \text{ for all $\bx \in \mathbb{R}^d \setminus \{\ba\}$}.
\end{equation}

We now check all the conditions required for well-behavedness as given in \cref{def:well-behaved}. Since $X'$ is a bounded polytope given as a system of linear inequalities, and since $\ba \notin X'$, we can compute in polynomial time a sufficiently small $\tau \in (0,1)$ such that $\min_{\bx \in X'} \|\bx - \ba\|_\infty \geq \tau$, and $\sz(\tau)$ is polynomial in the size of the representation of $X'$ and in $\sz(\ba)$.

We set $B = \lceil \log (\max\{1,q\}/\tau) \rceil$, $C = 4$, and $\beta_{min} = 2\tau$. Note that $B$ and $C$ have at most polynomial magnitude in the size of the input, and $\sz(\beta_{min})$ is also polynomial in the size of the input.
We have the following:
\begin{itemize}
    \item Let us first check the condition that needs to be satisfied by $\beta_{min}$. Consider any $\bx \in X'$. By construction of $\tau$, there must be a $j \in [d]$ such that $|x_j - a_{j}| \ge \tau$. So, $\sum_{j \in [d]} (x_j - a_{j})^2 \ge \tau^2$. From \eqref{eq:lm:one-charge-well-behaved:3}, for all $\bx \in X'$, we have
    \[
        M^{(k)}_f(\bx) \le \frac{k! 2^k q}{( \sum_{j \in [d]} (x_j - a_{j})^2 )^{(k+1)/2}} \le \frac{k! 2^k q}{\tau^{k+1}} = \frac{k! 4^k q}{ \tau (2\tau)^k} \leq C^k 2^B \frac{k!}{\beta_{min}^k},
    \]
    as required.
    \item Let us now pick an arbitrary $\beta \ge \beta_{min}$. If we have 
    \begin{align*}
        M^{(k)}_f(\bx) &> \frac{C^k 2^B k!}{\beta^k} \geq \frac{4^k 2^{\log(q/\tau)} k!}{\beta^k} = \frac{k! 2^k q}{\tau (\beta/2)^k},
    \end{align*}
    then using \eqref{eq:lm:one-charge-well-behaved:3} we get
    \begin{align*}
        &\frac{k! 2^k q}{\tau (\beta/2)^k} < M^{(k)}_f(\bx) \le \frac{k! 2^k q}{( \sum_{j \in [d]} (x_j - a_{j})^2 )^{(k+1)/2}} \le \frac{k! 2^k q}{\tau ( \sum_{j \in [d]} (x_j - a_{j})^2 )^{k/2}} \\
        &\implies ( \sum_{j \in [d]} (x_j - a_{j})^2 )^{k/2} \le (\beta/2)^k \implies \sum_{j \in [d]} (x_j - a_{j})^2 \le (\beta/2)^2,
    \end{align*}
    which implies $|x_j - a_{j}| \le \beta/2$ for all $j \in [d]$. So, we can set $n_{\beta} = 1$ (trivially polynomially bounded in the input size) and the corresponding hypercuboid that covers this region to be the hypercube $\bigtimes_{j \in [d]} [a_j - \beta/2, a_j + \beta/2]$, which has all sides of length at most $\beta$ as required.
\end{itemize}
\end{proof}

\section{Computing Strong Approximate Solutions}\label{sec:strong}

In \cref{sec:weak}, we presented an algorithm for computing approximate stationary points of an electrostatic potential.
However, these points are not guaranteed to lie close to exact solutions. In this section, we show how to compute \emph{strong} approximate stationary points, namely points that are guaranteed to lie close to exact solutions. Our algorithm requires an additional assumption, that we call \textit{strong non-degeneracy}.

A point $\bx$ is a degenerate stationary point of $f$ if $(\nabla f)(\bx) = 0$ and $(\nabla^2 f)(\bx) = 0$, i.e., both the gradient and the Hessian of $f$ are $0$ at $\bx$.\footnote{Note that $\nabla^2$ is the Hessian $\nabla \times \nabla$ operator and not the Laplacian $\nabla \cdot \nabla$ operator. The Laplacian of the real-life electrostatic potential is $0$ everywhere, but the Hessian is not (for non-trivial instances).} 
On the other hand, $\bx$ is a non-degenerate stationary point if $(\nabla f)(\bx) = 0$ but $(\nabla^2 f)(\bx) \neq 0$.
A function $f$ is non-degenerate if all its stationary points are non-degenerate.
Most \textit{generic} functions are non-degenerate, i.e., parameterized functions are non-degenerate for \textit{almost all} instantiations of the parameters (degenerate only for a measure $0$ subspace of the parameters).
For example, among all degree-$k$ polynomials, parameterized by the coefficients of the monomials, the polynomials that are degenerate correspond to a subset of the coefficients of measure $0$. Another example, the electrostatic potential function in \eqref{eq:potential} is degenerate only for a measure $0$ subspace of the possible locations and strengths of the charged particles. In other words, if we slightly perturb (even the location or strength of only one of) the charges, the electrostatic potential will \textit{almost surely} be non-degenerate~\cite[Chapter 6, 32]{morse2014critical}.

We assume a stronger form of non-degeneracy in our analysis: 
We call a function $f$ $\delta$-strongly non-degenerate in a domain $X$ if there exists a stationary point $\bx \in X$ where the determinant of the Hessian at $\bx$ has magnitude at least $\delta > 0$, i.e., $(\nabla f)(\bx) = 0$ and $|\det((\nabla^2 f)(\bx))| \ge \delta$.
With this assumption, we provide an algorithm to compute a strong $\epsilon$-stationary point of $f$ that runs in time polynomial in $\log(1/\epsilon)$ and $\log(1/\delta)$, so we could set $\epsilon$ and $\delta$ exponentially small.

\begin{theorem}\label{thm:strong-algo-no-smoothed}
Fix any $d \geq 1$ and any well-behaved family $\mathcal{F}$ of efficiently representable smooth functions in $d$-dimensional space. There exists an algorithm that takes as input
\begin{itemize}
    \item a function $f \in \mathcal{F}$ (given by its representation $r$, and defined on domain $X_r \subset \bR^d$),
    \item an error parameter $\eps > 0$,
    \item and a bounded domain $X \subset \bR^d$ (given as a system of linear inequalities) with $X \subseteq X_r$,
\end{itemize}
and which has the following guarantee. Whenever there exist $\bx^* \in X$ and $\delta > 0$ such that $(\nabla f)(\bx^*) = 0$ and $| \det((\nabla^2 f)(\bx^*)) | \ge \delta$, the algorithm runs in time polynomial in $|r|$, $\log(1/\eps)$, $\log(1/\delta)$, and the size of the representation of $X$, and outputs a point $\bx \in X$ such that there exists $\bx' \in \mathbb{R}^d$ with $\|\bx-\bx'\| \leq \eps$ and $\nabla f(\bx') = 0$.
\end{theorem}

The rest of this section provides a proof of the theorem. We will assume that we are given $\delta > 0$ that satisfies the condition in the theorem. If we are not given $\delta$, we can run the algorithm that we present here for values of $\delta$ decreasing exponentially $1, 1/2, 1/4, 1/8, \ldots$ and obtain the desired result.

As the family $\mathcal{F}$ is well-behaved, using \cref{lm:derivatives}, $\frac{\partial f}{\partial x_j}$ and $\frac{\partial^2 f}{\partial x_j \partial x_{j'}}$ also satisfy the well-behavedness condition for all $j, j' \in [d]$. Further, using \cref{lm:linear-comb,lm:product-comb}, $\det(\nabla^2 f)$ is also well-behaved.

As a result, we can use \cref{thm:well-behaved-easy} to solve the following system of inequalities over $X$, where $\eps' > 0$ is some sufficiently small value that we fix later.
\begin{align}
\begin{aligned}
    \label{eq:app:strong} 
    \frac{\partial f}{\partial x_j}(\bx) \leq \eps' \quad &\text{ and } \quad \frac{\partial f}{\partial x_j}(\bx) \geq -\eps' ,& \qquad\text{ for $j \in [d]$,}\\
    \det(\nabla^2 f(\bx)) \ge \delta - \eps' \quad &\text{ or }\quad \det(\nabla^2 f(\bx)) \le -\delta + \eps'.&
\end{aligned}
\end{align}
The ``or'' condition among the last two inequalities can be implemented by solving the system of inequalities once with the first condition, and once with the second condition.
Notice that the point $\bx^* \in X$ satisfies the above inequalities with at least $\eps'$ slack. As a result, using \cref{thm:well-behaved-easy} with error parameter $\eps'$, we obtain a point $x \in X$ that satisfies the system of inequalities \eqref{eq:app:strong}.

Next, we show that a point $\bx \in X$ that satisfies \eqref{eq:app:strong} is an $\eps$-strong approximate equilibrium.

\begin{lemma}\label{lm:strong}
Let $B$, $C$, and $\beta_{min}$ be the parameters associated with the well-behavedness property of $f$ as introduced in \cref{def:well-behaved}.
Let $\delta'$, $\alpha$, and $\epsilon'$ be defined as
\begin{align*}
    \delta' = \frac{\delta}{2} \left( \frac{\beta_{min}^{2}}{2 d C^2 2^B} \right)^{d-1},\qquad
    \alpha = \min\left( \frac{\epsilon}{\sqrt{d}}, \frac{\delta' \beta_{min}^3}{8 d^2 C^3 2^B} \right), \qquad
    \epsilon' = \frac{3}{16} \delta' \alpha.
\end{align*}
Let $\bx \in X$ be a point that satisfies \eqref{eq:app:strong}. We claim that there exists a point $\bu \in \bR^d$ such that $(\nabla f)(\bx + \bu) = 0$ and $|| \bu ||_{\infty} \le \epsilon/\sqrt{d} \implies || \bu || \le \epsilon$.
\end{lemma}
\begin{proof}
Let $A = (\nabla^2 f)(\bx)$. We have $|\det(A)| \ge \delta - \epsilon' \ge \delta/2$.
Let $(\lambda_j)_{j \in [d]}$ be the eigenvalues of $A$, and let us reindex these eigenvalues as $0 < |\lambda_1| \le |\lambda_2| \le \ldots \le |\lambda_d|$ without loss of generality. 
As $f$ satisfies the well-behavedness property, each element in the matrix $A$ is bounded by $M^{(2)}_f(\bx) \le C^2 2^B \frac{2}{\beta_{min}^2}$; therefore, $|\lambda_j| \le d \cdot M^{(2)}_f(\bx) \le d C^2 2^B \frac{2}{\beta_{min}^2}$ for all $j \in [d]$.
As a result, we have
\begin{align*}
    |\lambda_1| = \frac{| \det(A) |}{\prod_{j > 1} | \lambda_j |} \ge \frac{\delta}{2} \left( \frac{\beta_{min}^{2}}{2 d C^2 2^B} \right)^{d-1} = \delta'.
\end{align*}

Let $h : \bR^d \to \bR^d$ be defined as $h(\bu) = ( h_j(\bu) )_{j \in [d]} = A^{-1} (\nabla f) (\bx + \bu) $. Let us do a second-order Taylor expansion of $h$ with respect to $0$
\begin{align*}
    h(\bu) &= h(0) +  ((\nabla h) (0))^\intercal \bu +  \left( \frac{1}{2} \sum_{k, k' \in [d]} u_k u_{k'} \frac{\partial^2 h_j}{\partial u_k u_{k'} } (t \bu) \right)_{j \in [d]} ,
\end{align*}
for some $t \in [0,1]$. Let us look at each term in the above Taylor expansion separately. 
Recall that we have $|| (\nabla f)(\bx) ||_\infty \le \epsilon'$. Therefore, we can bound the first term, $h(0)$, as
\begin{align*}
    || h(0) ||_\infty = || A^{-1} (\nabla f)(\bx)  ||_\infty \le \frac{\epsilon'}{|\lambda_1|} \le \frac{\epsilon'}{\delta'} \le \frac{3 \alpha}{8},
\end{align*}
where the last inequality is by the definition of $\epsilon'$.
Now, the second term, $((\nabla h) (0))^\intercal \bu$, is simply
\begin{align*}
    ((\nabla h) (0))^\intercal \bu = ( A^{-1} (\nabla^2 f)(\bx) )^\intercal \bu = \bu.
\end{align*}
Finally, let us look at the third term. Notice that 
\[
    \left| \frac{\partial^2 h_j}{\partial u_k u_{k'} } (t \bu) \right| \le \frac{1}{|\lambda_1|} \max_{\by \in X, \ell \in [d]} \left| \frac{\partial^3 f}{\partial x_{\ell} x_k x_{k'} } (\by) \right| \le \frac{1}{|\lambda_1|} \max_{\by \in X} M^{(3)}_f(\by) \le \frac{C^3 2^B 6}{\delta' \beta_{min}^3}.
\]
Therefore, each element in the length-$d$ vector of the third term is bounded by
\begin{align*}
    \left| \frac{1}{2} \sum_{k, k' \in [d]} u_k u_{k'} \frac{\partial^2 h_j}{\partial u_k u_{k'} } (t \bu) \right| \le \frac{3 d^2 || \bu ||_{\infty}^2 C^3 2^B }{\delta' \beta_{min}^3} \le \frac{3 || \bu ||_{\infty}^2}{8 \alpha},
\end{align*}
where the last inequality is by the definition of $\alpha$.
Plugging the bounds for the individual terms back into the Taylor expansion, we get
\begin{align}\label{eq:lm:strong:1}
    | h_j(\bu) - u_j | \le \frac{3 \alpha}{8} + \frac{3 || \bu ||_{\infty}^2}{8 \alpha}.
\end{align}

We will now use the Poincar\'e--Miranda theorem (the multivariate generalization of the intermediate value theorem) stated below to complete the proof.
\begin{theorem}[Poincar\'e--Miranda]
Let $f = (f_1, f_2, \dots, f_n) : [-a, a]^n \to \bR^n $ be a continuous function such that, for each $j \in [n]$, the following two conditions hold:
\begin{align*}
f_i(x_1, \dots, x_{i-1}, -a, x_{i+1}, \dots, x_n) &\leq 0 \text{ for all } (x_1, \dots, x_{i-1}, x_{i+1}, \dots, x_n) \in [-a, a]^{n-1}, \\
f_i(x_1, \dots, x_{i-1}, a, x_{i+1}, \dots, x_n) &\geq 0 \text{ for all } (x_1, \dots, x_{i-1}, x_{i+1}, \dots, x_n) \in [-a, a]^{n-1}.
\end{align*}
Then there exists a point $ \bx \in [-a, a]^n $ such that $f(\bx) = 0$.
\end{theorem}

Let us consider the hypercube $U = [-\alpha, \alpha]^d$. 
Notice that for any $\bu \in U$, we have $|| \bu ||_{\infty} \le \alpha$. So, from \eqref{eq:lm:strong:1}, we have
\begin{align*}
    | h_j(\bu) - u_j | \le \frac{3 \alpha}{8} + \frac{3 || \bu ||_{\infty}^2}{8 \alpha} \le \frac{3 \alpha}{4} \implies \begin{cases}
        h_j(\bu) \ge u_j - \frac{3 \alpha}{4} = \frac{\alpha}{4} \ge 0, &\text{ if $u_j = \alpha$,}\\
        h_j(\bu) \le u_j + \frac{3 \alpha}{4} = -\frac{\alpha}{4} \le 0, &\text{ if $u_j = -\alpha$.}
    \end{cases}
\end{align*}
As $\alpha \le \epsilon/\sqrt{d}$, applying the Poincar\'e--Miranda theorem completes the proof.
\end{proof}

\section{Computational Complexity of Generalizations}\label{sec:hardness}

In \cref{sec:weak}, we gave an efficient algorithm for computing stationary points of well-behaved functions in a constant dimensional space. There are two natural generalizations: First, we can consider well-behaved functions in a high-dimensional space. Second, we can restrict to constant dimensions but consider functions that are not well-behaved. 
We study these two generalizations in this section.

\subsection{Well-Behaved Functions in a High-Dimensional Space}

In this section we show that our approach cannot be generalized to settings where the dimension of the space is not fixed, but can grow polynomially. Namely, we prove that computing an approximate stationary point of a well-behaved function is \np/-hard in higher dimension.

The class of functions that we consider are simply the polynomials of degree 3 over $d$ variables, for all $d \geq 1$. This class of functions is well-behaved, as shown in \cref{app:well-behaved-properties}.
Ahmadi and Zhang~\cite{AhmadiZ22-unconstrained} have shown that it is \np/-hard to decide whether such a polynomial admits a stationary point over $\mathbb{R}^d$. We use their reduction to show a stronger result, namely that a so-called gap version of the problem remains hard, which implies that our result for fixed dimension cannot be generalized to higher dimension.

\begin{theorem}\label{thm:well-behaved-high-dim}
Given a polynomial $p: \mathbb{R}^d \to \mathbb{R}$ of degree $3$, it is \np/-hard to distinguish between the following two cases:
\begin{itemize}
    \item $p$ has a stationary point lying in $[-1,1]^d$,
    \item $p$ has no $1/d^2$-stationary point in $\mathbb{R}^d$.
\end{itemize}
\end{theorem}

\begin{proof}
We closely follow the reduction used by Ahmadi and Zhang~\cite[Theorem~2.1]{AhmadiZ22-unconstrained}. They reduce from the problem of deciding whether an unweighted and undirected graph on $n$ vertices admits a cut of size $k$. This is equivalent to deciding whether the following system of quadratic equations is feasible.
\begin{equation}\label{eq:max-cut-exact}
\begin{split}
q_0(x) &= \frac{1}{4} \sum_{i=1}^n \sum_{j=1}^n E_{ij} (1 - x_i x_j) - k = 0,\\
q_i(x) &= x_i^2 - 1 = 0, \qquad i=1, \dots, n,
\end{split}
\end{equation}
where $E_{ij} \in \{0,1\}$ denotes whether or not the edge $\{i,j\}$ is present in the graph. Indeed, the constraints $q_i$ force the variables $x_i$ to be $-1$ or $+1$, and this can then be interpreted as assigning each corresponding vertex to one or the other side of the cut. It is then easy to check that $q_0$ is equal to zero if and only if the cut has size exactly $k$.\footnote{Note that every edge appears twice in the sum.}

Now, we observe that \eqref{eq:max-cut-exact} is feasible if and only if the following relaxed version of the system is feasible, where we set $\delta = 1/n^2$.
\begin{equation}\label{eq:max-cut-approx}
\begin{split}
q_0(x) &= \frac{1}{4} \sum_{i=1}^n \sum_{j=1}^n E_{ij} (1 - x_i x_j) - k \in [-\delta, \delta],\\
q_i(x) &= x_i^2 - 1 \in [-\delta, \delta], \qquad i=1, \dots, n,
\end{split}
\end{equation}
Clearly, if \eqref{eq:max-cut-exact} is feasible, then so is \eqref{eq:max-cut-approx}. Now, assume that \eqref{eq:max-cut-approx} is feasible, and let $x \in \mathbb{R}^n$ denote a feasible point. Then, $x_i^2-1 \in [-\delta, \delta]$, implies that $x_i \in [\sqrt{1-\delta}, \sqrt{1+\delta}]$ or $x_i \in [-\sqrt{1+\delta}, -\sqrt{1-\delta}]$. We define the rounded version of $x$ to be the vector $\hat{x} \in \mathbb{R}^n$ given by $\hat{x}_i = 1$ if $x_i \in [\sqrt{1-\delta}, \sqrt{1+\delta}]$, and $\hat{x}_i = -1$ if $x_i \in [-\sqrt{1+\delta}, -\sqrt{1-\delta}]$. We now show that $\hat{x}$ is feasible for \eqref{eq:max-cut-exact}. Indeed,
\begin{equation*}
|q_0(\hat{x}) - q_0(x)| \leq \frac{1}{4} \sum_{i=1}^n \sum_{j=1}^n E_{ij} |x_i x_j - \hat{x}_i \hat{x}_j| \leq \frac{1}{4} \sum_{i=1}^n \sum_{j=1}^n E_{ij} \cdot \delta \leq 1/4
\end{equation*}
because $\delta = 1/n^2$. Now, since $|q_0(x)| \leq \delta < 3/4$ and $q_0(\hat{x})$ is an integer, it follows that $q_0(\hat{x}) = 0$. Thus, \eqref{eq:max-cut-exact} is also feasible.

Finally, just as in the reduction of Ahmadi and Zhang~\cite[Theorem~2.1]{AhmadiZ22-unconstrained}, we construct the polynomial
$$p(x,y) = \sum_{i=0}^n y_i q_i(x)$$
over the $d=2n+1$ variables $(x_1, \dots, x_n, y_0, y_1, \dots, y_n)$. We thus have $\partial p/\partial x_j(x,y) = \sum_{i=0}^n y_i \partial q_i / \partial x_j(x)$ and $\partial p/\partial y_j (x,y) = q_j(x)$.

Now, if \eqref{eq:max-cut-exact} is feasible, letting $x \in [-1,1]^n$ denote a feasible point for it, then the point $(x,0, \dots, 0) \in [-1,1]^{2n+1}$ is stationary for $p$. It remains to show that if \eqref{eq:max-cut-exact} is not feasible, then $p$ has no $1/d^2$-stationary point in $\mathbb{R}^d$. We prove the contrapositive. Let $(x,y)$ be a $1/d^2$-stationary point of $p$. Then, since $1/d^2 \leq 1/n^2 = \delta$, $x$ is feasible for \eqref{eq:max-cut-approx}. But, as proved above, this implies that \eqref{eq:max-cut-exact} is also feasible. This completes the proof.
\end{proof}

\subsection{Non-Well-Behaved Functions in a Constant-Dimensional Space}

In this section we prove \cls/-hardness of computing equilibrium points of a generalized potential function in a two-dimensional space with only two charges. Each charge creates a potential that is broadly similar to the electrostatic potential. In particular, the potential function due to a given charge is continuously differentiable, has a singularity at the location of the charge, and is monotonically decreasing as we move away from the charge. The overall potential due to the two charges is the sum of the potentials due to each charge, as with the electrostatic potential.

\begin{theorem}\label{thm:cls-hard}
Given two locations $a_1,a_2 \in \mathbb{R}^2$ and two potentials\footnote{The functions and their derivatives are given by arithmetic circuits, see~\cite{FGHS22}.} $f_1,f_2: \mathbb{R}^2 \to \mathbb{R}$ that satisfy for $i \in \{1,2\}$
\begin{itemize}
    \item $f_i$ has a singularity at location $a_i$, but is continuously differentiable on $\mathbb{R}^2\setminus\{a_i\}$,
    \item $f_i$ is monotonically decreasing as we move away from the location $a_i$,
\end{itemize}
it is \cls/-hard to compute a weak $\varepsilon$-approximate equilibrium point of $f_1+f_2$.
\end{theorem}

We prove our hardness result by a reduction from the \textsc{$2$D-min-max} problem defined below, which is known to be \cls/-hard~\cite{KalavasisPSZ25-2D-minmax}.

\begin{definition}\label{def:2D-min-max}
In the \textsc{$2$D-min-max} problem we are given a continuously differentiable function\footnote{The function and its derivative are again given by arithmetic circuits.} $f: [-1,1]^2 \to \mathbb{R}$, an approximation factor $\varepsilon > 0$, and a value $L > 0$, such that $\nabla f$ is $L$-Lipschitz-continuous over $[-1,1]^2$.
The goal is to find an $\varepsilon$-KKT point of the optimization problem
$$\min_{x_1 \in [-1,1]} \max_{x_2 \in [-1,1]} f(x_1,x_2)$$
i.e., a point $(x_1,x_2) \in [-1,1]^2$ that satisfies
\begin{itemize}
    \item if $x_1 \neq -1$, then $\frac{\partial f}{\partial x_1} \leq \varepsilon$,
    \item if $x_1 \neq 1$, then $\frac{\partial f}{\partial x_1} \geq - \varepsilon$,
    \item if $x_2 \neq -1$, then $\frac{\partial f}{\partial x_2} \geq - \varepsilon$,
    \item if $x_2 \neq 1$, then $\frac{\partial f}{\partial x_2} \leq \varepsilon$.
\end{itemize}
\end{definition}

\begin{theorem}[\cite{KalavasisPSZ25-2D-minmax}]\label{thm:minmax}
The \textsc{$2$D-min-max} problem is \cls/-hard. Further, this hardness result holds even if we assume that the function does not have an $\varepsilon$-KKT point at the boundary, and additionally $|f| \le 1$ and $||\nabla f||_2 \le 1$ over $[-1,1]^2$.
\end{theorem}

\begin{figure}[t]
    \begin{subfigure}{.5\textwidth}
        \centering
        \includegraphics[width=\linewidth]{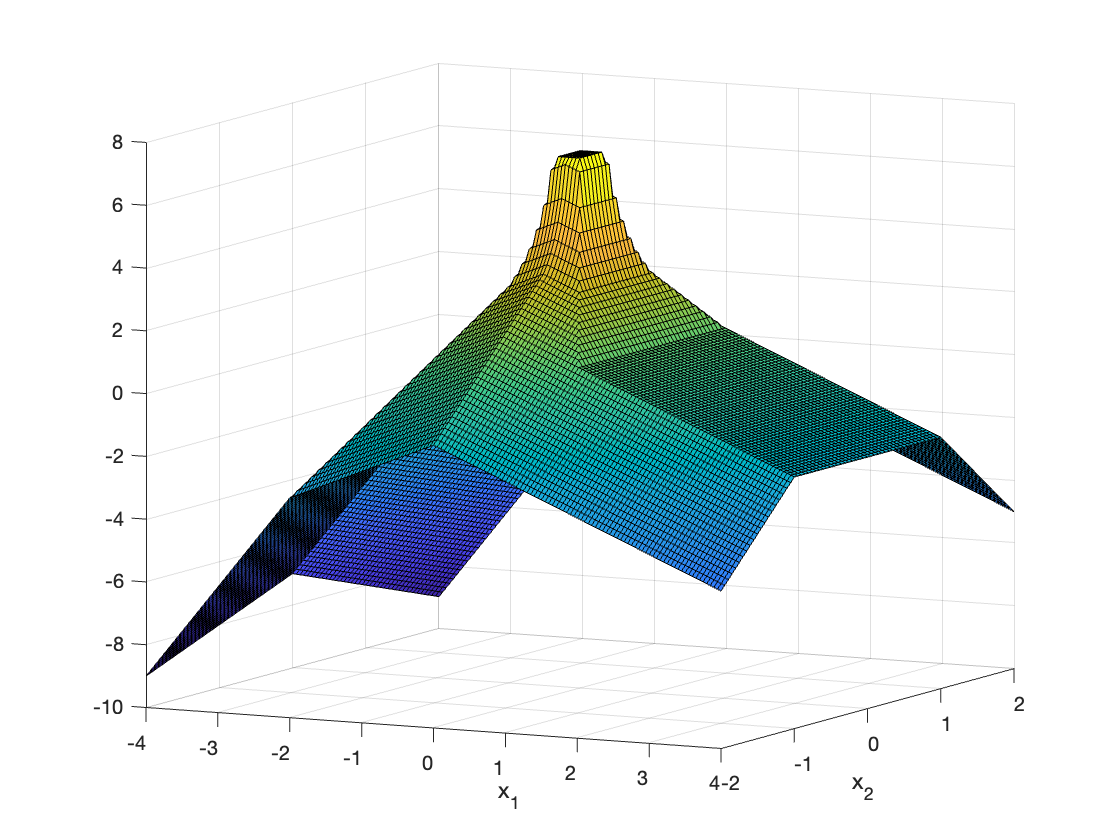}
        \caption{Potential induced by a single charge.\\ This corresponds to function $h$ in the proof.\\}
        \label{fig:pointy-one}
    \end{subfigure}%
    \begin{subfigure}{.5\textwidth}
        \centering
        \includegraphics[width=\linewidth]{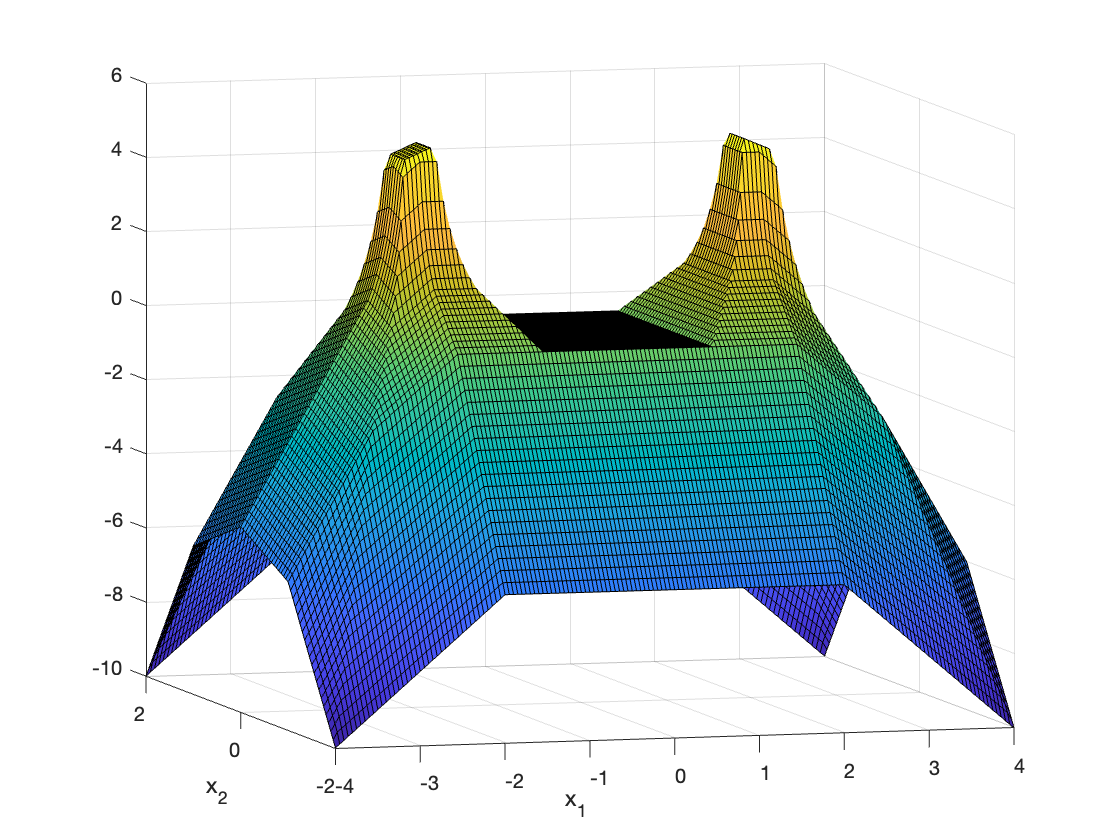}
        \caption{Potential induced by two charges.\\ This corresponds to the net potential due to the two charges, assuming $\hf = 0$ in the proof.}
        \label{fig:pointy-two}
    \end{subfigure}
    \caption{Generalized potential function due to two charges (non-smoothed).}
    \label{fig:pointy}
\end{figure}

Before proceeding with the proof of \cref{thm:cls-hard}, we give a brief overview of our approach. We first construct two piecewise linear potentials that are decreasing around their respective charge locations. These are carefully constructed such that, when they are added together, a flat region appears between the two singularities. See \cref{fig:pointy-one} for the potential associated to a single charge, and \cref{fig:pointy-two} for the sum of the two potentials and the flat region.

Next, we modify one of the potentials very carefully, such as to embed an instance of the \textsc{$2$D-min-max} problem in that flat region. As a result, we argue that any stationary point must correspond to a KKT point of the \textsc{$2$D-min-max} instance. Importantly, we make sure that our two potentials remains monotonically decreasing, despite the modifications. Finally, we apply a box filter to make our two potentials continuously differentiable.

\begin{proof}[Proof of \cref{thm:cls-hard}]
We reduce from the \textsc{$2$D-min-max} problem (without any boundary solution, as stated in \cref{thm:minmax}) to computing an equilibrium point of the generalized electrostatic potential. 
In particular, this implies that if we can compute a weak $\varepsilon$-approximate stationary point of the electrostatic potential in $\bR^2$ (excluding the locations of the charges) in $\poly(\log(1/\varepsilon))$ time, then we can compute an $\varepsilon$-KKT point of the \textsc{$2$D-min-max} problem in $\poly(\log(1/\varepsilon))$ time.

We are given $\eps > 0$ and a function $f: [-1,1]^2 \to \bR$ as defined in \cref{def:2D-min-max}. $\nabla f$ is $L$-Lipschitz-continuous over $[-1,1]^2$. $f$ and its derivative are bounded, i.e., $|f| \le 1$ and $||\nabla f||_2 \le 1$. We also assume that there is no $\varepsilon$-KKT point of $f$ on the boundary of $[-1, 1]^2$.

We construct the generalized electrostatic potential using the following function $h : \bR^2 \to \bR$ as a building block. The function $h$ is defined as
\begin{align*}
    h(x_1, x_2) = \begin{cases}
        2 + 1/|x_1| - |x_1| - |x_2|,& \text{if $x_1 \le |x_2|$ and $|x_2| \le |x_1| \le 1/2$},\\
        2 + 1/|x_2| - |x_1| - |x_2|,& \text{if $x_1 \le |x_2|$ and $|x_1| \le |x_2| \le 1/2$},\\
        5 - 3|x_1| - |x_2|,& \text{if $x_1 \le |x_2|$, $\max(|x_1|,|x_2|) \ge 1/2$, and $|x_2| \le |x_1|$},\\
        5 - |x_1| - 3|x_2|,& \text{if $x_1 \le |x_2|$, $\max(|x_1|,|x_2|) \ge 1/2$, and $|x_1| \le |x_2|$},\\
        2 + 1/|x_1| - 2|x_1|,& \text{if $x_1 \ge |x_2|$ and $x_1 \le 1/2$},\\
        5 - 4|x_1|,& \text{if $x_1 \ge |x_2|$ and $1/2 \le x_1 \le 1$},\\
        2 - |x_1|,& \text{if $x_1 \ge |x_2|$, $x_1 \ge 1$, and $|x_2| \le 1$},\\
        5 - |x_1| - 3|x_2|,& \text{if $x_1 \ge |x_2|$, $x_1 \ge 1$, and $|x_2| \ge 1$}.
    \end{cases}
\end{align*}
The function $h$ is plotted in \cref{fig:pointy-one} (we cut off $h$ at $h(x_1, x_2) \le 5$ in the plot). The last four cases above cover the right side $x_1 \ge |x_2|$ in the plot, while the first four cases cover the other three sides, which are similar. We can easily observe that $h$ decreases as we move away from $(0,0)$. We can also check (or observe from the plot) that $h$ is continuous.

Let us extend the function $f$ given in the min-max problem, which is only defined over the domain $[-1, 1]^2$, to all of $\bR^2$.
Let $\psi : \bR \to \bR$ be the function that clips a number between $[-1, 1]$, i.e., $\psi(z) = \max(-1, \min(1, z))$.
The extension of $f$ to $\bR^2$, denoted by $\hf$, is given as
\[
    \hf(\bx) = \hf(x_1, x_2) = f(\psi(x_1), \psi(x_2)).
\]
Note that $\hf$ is continuous on all of $\bR^2$. However, $\hf$ could be non-differentiable at points $\bx$ where $|x_1| = 1$ or $|x_2| = 1$, but it is continuously differentiable everywhere else.

We now use the functions $h$ and $\hf$ to define the potential functions of the two charges.
We place the two charges at $\ba_1 = (-2, 0)$ and $\ba_2 = (2, 0)$, respectively.
As an intermediate step towards defining the potential due to the charges, let us first define the following two functions
\begin{align*}
    f_1(\bx) = f_1(x_1, x_2) &= h(2 + x_1, x_2) + \hf(\bx)/2 \\
    f_2(\bx) = f_2(x_1, x_2) &= h(2 - x_1, x_2) + \hf(\bx)/2.
\end{align*}
For each $i \in [2]$, notice that $f_i$ is continuous in all of $\bR^2$ (except at the point of singularity $\ba_i$) because $h$ is continuous in all of $\bR$ except $(0,0)$ and $\hf$ is continuous everywhere. However, $f_i$ is not differentiable everywhere, because $h$ is not differentiable at the boundaries of the pieces defining it, and $\hf$ is not differentiable at points $\bx$ where $|x_1| = 1$ or $|x_2| = 1$. Despite this non-differentiability, we can check that $f_i$ is absolutely continuous. In particular, we can check from its definition that, fixing $x_1$, the function $f_i(x_1, \cdot) : \bR \to \bR$ is continuous everywhere (except at the singularity point $\ba_i$) and differentiable everywhere except at finitely many points. The same holds for $f_i(\cdot, x_2)$. So, we can apply the fundamental theorem of calculus to $f_i(x_1, \cdot)$ and $f_i(\cdot, x_2)$,\footnote{\url{https://en.wikipedia.org/wiki/Fundamental_theorem_of_calculus}} which we do later in the analysis.

We define the potential due to charge $i$, denoted by $g_i$, by applying a continuous \textit{box filter} (or box blur, or continuous moving average, or convolution with an indicator box kernel) on the function $f_i$.\footnote{\url{https://en.wikipedia.org/wiki/Convolution}, \url{https://en.wikipedia.org/wiki/Box_blur}.} 
Let $r = \frac{\varepsilon}{4L} > 0$. We apply the box filter on $f_i$ with a box of the form $[-r, r]^2$ on all\footnote{Actually, we need to do something a bit different when we are close to the singularity $\ba_i$. One way to handle this is to first replace $f_i$ by $\min\{f_i,M\}$, where $M>0$ is some sufficiently large value. Then, we can apply the box filter, since the singularity is no longer there. Finally, we reintroduce the singularity by adding a continuously differentiable function $s$ to $f_i$ such that $s$ has a singularity at $\ba_i$ and $s(\bx) = 0$ when $\|\bx - \ba_i\|_2 \geq t$ for some appropriate $t > 0$. Such a function $s$ can for example be obtained by letting $s(\bx) = 1/\|\bx-\ba_i\|_2 + \|\bx-\ba_i\|_2/t^2 + -2/t$ when $\|\bx - \ba_i\|_2 \leq t$, and $s(\bx) = 0$ otherwise. This ensures that $s$ and thus $s+f_i$ is continuously differentiable. It is not hard to show that we have only modified the function in a region where there is no solution. We omit the details.} of $\bR^2$. We obtain $g_i$ as given below 
\[
    g_i(\bx) = \frac{1}{4 r^2} \int_{u_1 = x_1 - r}^{x_1 + r} \int_{u_2 = x_2 - r}^{x_2 + r} f_i(u_1, u_2) du_2 du_1.
\]
The net potential due to the two charges is $g = g_1 + g_2$. \cref{fig:pointy-two} plots $f_1 + f_2$, which is similar to $g$, except that we have not added $\hf$ (which adds small changes to the shape) and applied the box filter (which smoothens out the corners).

First, notice that $g$ is continuously differentiable. 
This property follows directly from using the box filter---if a function is continuous, then applying the box filter operation makes it continuously differentiable. In particular, applying the Leibniz integral rule, we get
\begin{align}\label{eq:thm:cls-hard:1}
    \frac{\partial g_i}{\partial x_1}(x_1, x_2) = \frac{1}{4r^2} \int_{u_2 = x_2 - r}^{x_2 + r} ( f_i(x_1 + r, u_2) - f_i(x_1 - r, u_2) ) du_2.
\end{align}
As $f_i$ is continuous, $\frac{\partial g_i}{\partial x_1}$ is also continuous using the expression above. Similarly, we can check that $\frac{\partial g_i}{\partial x_2}$ is also continuous. So, $g_1$, $g_2$, and $g = g_1 + g_2$ are continuously differentiable.

We now claim that $\frac{\partial g_i}{\partial x_j}(x_1, x_2)$ is an average of $\frac{\partial f_i}{\partial x_j}(u, v)$ for $(u, v) \in [x_1 - r, x_1 + r] \times [x_2 - r, x_2 + r]$. 
As discussed earlier, we know that $f_i$ is absolutely continuous.
In particular, fixing $u_2$, the function $f_i(\cdot, u_2) : \bR \to \bR$ is continuous everywhere and differentiable everywhere except finitely many points. Using the fundamental theorem of calculus, we have
\[
    f_i(x_1 + r, u_2) - f_i(x_1 - r, u_2) = \int_{u_1 = x_1 - r}^{x_1 + r} \frac{\partial f_i}{\partial x_1}(u_1, u_2) du_1.
\]
Plugging this in \eqref{eq:thm:cls-hard:1}, we get
\begin{align}\label{eq:thm:cls-hard:2}
    \frac{\partial g_i}{\partial x_1}(x_1, x_2) = \frac{1}{4r^2} \int_{u_2 = x_2 - r}^{x_2 + r} \int_{u_1 = x_1 - r}^{x_1 + r} \frac{\partial f_i}{\partial x_1}(u_1, u_2) du_1 du_2.
\end{align}
Analogously, we can show that 
\begin{align}\label{eq:thm:cls-hard:3}
    \frac{\partial g_i}{\partial x_2}(x_1, x_2) = \frac{1}{4r^2} \int_{u_2 = x_2 - r}^{x_2 + r} \int_{u_1 = x_1 - r}^{x_1 + r} \frac{\partial f_i}{\partial x_2}(u_1, u_2) du_1 du_2.
\end{align}

It is easy to check that $g_i$ is monotonically decreasing as we move away from charge $i$ located at $\ba_i$. 
Formally, we want to show that for any $\bx \in \bR^2 \setminus \{ \ba_i \}$, we have $(\bx - \ba_i)^\intercal ((\nabla g_i) (\bx)) < 0$.
From our construction, we know that $h$ decreases as we move away from $(0,0)$. Further, we can easily check---by considering the cases that define $h$---that $\nabla h$ has a directional component of least $1/\sqrt{2}$ pointing radially outward from $(0,0)$ almost everywhere. 
On the other hand, $\hf$ is continuously differentiable and has a gradient bounded by $||\nabla \hf||_2 \le ||\nabla f||_2 \le 1$ almost everywhere.
Therefore, $\nabla f_i$ has a directional component of least $1/\sqrt{2} - ||\nabla f||_2/2 \ge (\sqrt{2} - 1)/2$ pointing radially away from $\ba_i$ almost everywhere. As $\nabla g_i (\bx)$ is the average of $\nabla f_i$ in the $2r$-width box around $\bx$, so $\nabla g_i (\bx)$ is monotonically decreasing as we move away from $\ba_i$.

Next, we prove that there are no weak $(\varepsilon/2)$-approximate equilibrium points outside the box $[-1+r, 1-r]^2$. 
First, note that $g$ has a very high gradient very close to the singularities. 
In particular, by construction of $g$, in the $\bigtimes_{j \in [2]} [a_{i,j} - 1/2, a_{i,j} + 1/2]$ box around $\ba_i$, the potential $g_i$ due to charge $i$ has a gradient of magnitude at least $5 - || \nabla \hf ||_2/2$. 
On the other hand, the potential due to charge $(3-i)$ has a gradient of magnitude at most $1 + || \nabla \hf ||_2/2$ in this region. As $|| \nabla \hf ||_2/2 \le 1$, in this $\bigtimes_{j \in [2]} [a_{i,j} - 1/2, a_{i,j} + 1/2]$ box, the gradient of $g$ is at least $5 - || \nabla \hf ||_2/2 - (1 + || \nabla \hf ||_2/2) \ge 3$, and there can be no equilibrium point.

Let us now focus on points $\bx = (x_1, x_2)$ outside the $\bigtimes_{j \in [2]} [a_{i,j} - 1/2, a_{i,j} + 1/2]$ box around the singularity points $\ba_i$ for $i \in [2]$.
Note that $g = g_1 + g_2$ is the box filtered version of $f_1 + f_2$. 
Also, as argued earlier, $\nabla g (\bx)$ is the average of $\nabla (f_1 + f_2)$ in the $2r$-width box around $\bx$.
Further, $(f_1 + f_2)(x_1, x_2) = h(2 + x_1, x_2) + h(2 - x_1, x_2) + \hf(x_1, x_2)$. Let us consider the following two cases:
\begin{itemize}
    \item $\bx \notin [-1-r, 1+r]^2$. Using the definition of $h$, we explicitly write down $h(2 + x_1, x_2) + h(2 - x_1, x_2)$ below. We skip points in the $1$-width box around the singularities, which we have already considered earlier. We also write down the values for only the top-right quadrant, i.e., $x_1 \ge 0$ and $x_2 \ge 0$, as the function has reflection symmetry along both $x_1$-axis and $x_2$-axis. Let $\hh(x_1, x_2) = h(2 + x_1, x_2) + h(2 - x_1, x_2)$ for conciseness.
    \begin{align}
        &\hh(x_1, x_2) = h(2 + x_1, x_2) + h(2 - x_1, x_2) \nonumber \\
        &= \begin{cases}
            (skipped), &\text{if $x_1 \le 0$ or $x_2 \le 0$ [REGION-1]},\\
            (skipped), &\text{if $3/2 \le x_1 \le 5/2$ and $0 \le x_2 \le 1/2$ [REGION-2]},\\
            0, &\text{if $0 \le x_1 \le 1$ and $0 \le x_2 \le 1$ [REGION-3]},\\
            -3 + 3x_1, &\text{if $1 \le x_1 \le 3/2$, $0 \le x_2 \le 1$, and $x_1 + x_2 \le 2$ [REGION-4]},\\
            3 - 3x_2, &\text{if $1 \le x_1 \le 2$, $1/2 \le x_2 \le 1$, and $x_1 + x_2 \ge 2$ [REGION-5]},\\
            6 - 6x_2, &\text{if $0 \le x_1 \le 2$ and $x_2 \ge 1$ [REGION-6]},\\
            7 - 2x_1 - 3x_2, &\text{if $2 \le x_1 \le 3$, $1/2 \le x_2 \le 1$, and $x_1 \le x_2 + 2$ [REGION-7]},\\
            11 - 4x_1 - x_2, &\text{if $x_1 \ge 5/2$, $0 \le x_2 \le 1$, and $x_1 \ge x_2 + 2$ [REGION-8]},\\
            14 - 4x_1 - 4x_2, &\text{if $x_1 \ge 3$, $x_2 \ge 1$, and $x_1 \ge x_2 + 2$ [REGION-9]}\\
            10 - 2x_1 - 6x_2, &\text{if $x_1 \ge 2$, $x_2 \ge 1$, and $x_1 \le x_2 + 2$ [REGION-10]}.
        \end{cases} \label{eq:thm:cls-hard:4}
    \end{align}
    From the above definition, it is easy to check that the magnitude of the gradient of $h(2 + x_1, x_2) + h(2 - x_1, x_2)$ is at least $3$ almost everywhere outside the $[-1, 1]^2$ box.
    Further, $|| \nabla \hf ||_2 \le 1$ almost everywhere. So, the gradient of $f_1 + f_2$ has magnitude at least $2$, $|| \nabla (f_1 + f_2) ||_2 \ge 2$, almost everywhere outside the $[-1, 1]^2$ box.
    Therefore, $g$ will have a gradient at least $2$ outside the $[-1-r, 1+r]^2$ box.

    \item $\bx \in [-1-r, 1+r]^2 \setminus [-1+r, 1-r]^2$. That is, $\bx$ is in the small band of width $2r$ at the boundary of $[-1, 1]^2$. Let us focus on the top-right quadrant; a symmetric argument will apply to the other three quadrants as well. We consider the following three cases:

    CASE-1: $0 \le x_1 \le 1-r$ and $1-r \le x_2 \le 1+r$. We know that the min-max problem does not have an $\varepsilon$-KKT point at the boundary. So, when $0 \le x_1 < 1$ and $x_2 = 1$, either $|\frac{\partial f}{\partial x_1}| > \varepsilon$ or $\frac{\partial f}{\partial x_2} < -\varepsilon$. By $L$-Lipschitz-continuity of $\nabla f$, the definition of $\hf$, and as $r = \frac{\varepsilon}{4L}$, we have $|\frac{\partial \hf}{\partial x_1}| > \varepsilon - rL \ge \frac{\varepsilon}{2}$ or $\frac{\partial \hf}{\partial x_2} < -\varepsilon + rL \le -\frac{\varepsilon}{2}$ for $0 \le x_1 \le 1-r$ and $1-r \le x_2 \le 1+r$. Notice that when we apply the box filter to get $g$, we average $\hf$ with REGION-6 of $\hh$. In REGION-6, $\nabla \hh = (0, -6)$. Therefore, either $|\frac{\partial g}{\partial x_1}| = |\frac{\partial \hf}{\partial x_1}| > \frac{\varepsilon}{2}$ or $\frac{\partial g}{\partial x_2} \le \frac{\partial \hf}{\partial x_2} < -\frac{\varepsilon}{2}$ for $0 \le x_1 \le 1-r$ and $1-r \le x_2 \le 1+r$, as required.

    CASE-2: $1-r \le x_1 \le 1+r$ and $0 \le x_2 \le 1-r$. The argument is symmetric to the previous case.
    As the min-max problem does not have an $\varepsilon$-KKT point at the boundary, so either $\frac{\partial f}{\partial x_1} > \varepsilon$ or $|\frac{\partial f}{\partial x_2}| > \varepsilon$ when $x_1 = 1$ and $0 \le x_2 < 1$. By $L$-Lipschitz-continuity of $\nabla f$, the definition of $\hf$, and as $r = \frac{\varepsilon}{4L}$, we have $\frac{\partial \hf}{\partial x_1} > \varepsilon - rL \ge \frac{\varepsilon}{2}$ or $|\frac{\partial \hf}{\partial x_2}| > \varepsilon - rL \ge \frac{\varepsilon}{2}$ for $1-r \le x_1 \le 1+r$ and $0 \le x_2 \le 1-r$. Notice that when we apply the box filter to get $g$, we average $\hf$ with REGION-4 of $\hh$. In REGION-4, $\nabla \hh = (3, 0)$. Therefore, either $\frac{\partial g}{\partial x_1} \ge \frac{\partial \hf}{\partial x_1} > \frac{\varepsilon}{2}$ or $|\frac{\partial g}{\partial x_2}| = |\frac{\partial \hf}{\partial x_2}| > \frac{\varepsilon}{2}$ for $1-r \le x_1 \le 1+r$ and $0 \le x_2 \le 1-r$, as required.

    CASE-3: $1-r \le x_1 \le 1+r$ and $1-r \le x_2 \le 1+r$.
    Again, as the min-max problem does not have an $\varepsilon$-KKT point at the boundary, so either $\frac{\partial f}{\partial x_1} > \varepsilon$ or $\frac{\partial f}{\partial x_2} < -\varepsilon$ when $x_1 = 1$ and $x_2 = 1$. By $L$-Lipschitz-continuity of $\nabla f$, the definition of $\hf$, and as $r = \frac{\varepsilon}{4L}$, we have $\frac{\partial \hf}{\partial x_1} > \varepsilon - rL \ge \frac{\varepsilon}{2}$ or $\frac{\partial \hf}{\partial x_2} < -\varepsilon + rL < -\frac{\varepsilon}{2}$ for $1-r \le x_1 \le 1+r$ and $1-r \le x_2 \le 1+r$. Notice that when we apply the box filter to get $g$, we average $\hf$ with REGION-4, REGION-5, or REGION-6 of $\hh$. In these regions, $\nabla \hh$ is equal to $(3, 0)$, $(0, -2)$, or $(0 -6)$. Therefore, either $\frac{\partial g}{\partial x_1} \ge \frac{\partial \hf}{\partial x_1} > \frac{\varepsilon}{2}$ or $\frac{\partial g}{\partial x_2} \le \frac{\partial \hf}{\partial x_2} < -\frac{\varepsilon}{2}$ for $1-r \le x_1 \le 1+r$ and $1-r \le x_2 \le 1+r$, as required.
\end{itemize}

The above case analysis shows that there are no weak $(\varepsilon/2)$-approximate equilibrium points of $g$ outside the box $[-1+r, 1-r]^2$.

Finally, we claim that if a point $\bx \in [-1+r, 1-r]^2$ is an $(\varepsilon/2)$-approximate equilibrium point of $g$, then it must be an $\varepsilon$-KKT point of $f$. Notice that inside the $[-1, 1]^2$ box, $f_1 + f_2 = f$. Therefore, for any point $\bx \in [-1+r, 1-r]^2$, $(\nabla g) (\bx)$ is the average of $\nabla f$ in the $[-r, r]^2$ box around $\bx$. If $\bx$ is not an $\varepsilon$-KKT point of $f$, then either $|\frac{\partial f}{\partial x_1}(\bx)| > \varepsilon$ or $|\frac{\partial f}{\partial x_2}(\bx)| > \varepsilon$. Say $\frac{\partial f}{\partial x_1}(\bx) > \varepsilon$; the other cases are similar. By $L$-Lipschitz-continuity of $\nabla f$, for every point $\by$ in the $[-r,r]^2$ box around $\bx$, $\frac{\partial f}{\partial x_1}(\by) > \varepsilon - 2rL = \frac{\varepsilon}{2}$ as $r = \frac{\varepsilon}{4L}$. As $(\nabla g)(\bx)$ is the average of $(\nabla f)(\by)$ for $(\by - \bx) \in [-r, r]^2$, and as $\frac{\partial f}{\partial x_1}(\by) > \frac{\varepsilon}{2}$ for all such $\by$, we get $\frac{\partial g}{\partial x_1} (\bx) > \frac{\varepsilon}{2}$. So, $\bx$ cannot be a weak $(\varepsilon/2)$-approximate equilibrium point of $g$.
\end{proof}

\section*{Acknowledgments}
The authors wish to thank the reviewers for suggestions that helped improve the presentation of the paper. A.G.\ and P.W.G.\ were supported by EPSRC Grant EP/X040461/1 ``Optimisation for Game Theory and Machine Learning''.

\section*{Appendix}

\appendix

\section{Properties of Well-Behaved Functions}\label{app:well-behaved-properties}

In this section, we show that the well-behavedness property extends to derivatives, linear combinations, and products of well-behaved functions.

\begin{lemma}\label{lm:derivatives}
Let $\kappa \in \mathbb{N}$ be a constant. For any well-behaved family $\mathcal{F} = (f_r)_r$, the family $\mathcal{F}' = (D^{\bs} f_r)_{r,|\bs|=\kappa}$ is also well-behaved.
\end{lemma}
\begin{proof}
Let $f = f_r$ for some $r$, and consider $g = D^{\bs} f$, where $|\bs| = \kappa \ge 1$.

Let $B$, $C$, $(n_{\beta})_{\beta}$, and $\beta_{min}$ be the parameters associated with the well-behavedness property of the function $f$ as given in \cref{def:well-behaved}. Assume, without loss of generality that $\beta_{min} \leq 1$; if $\beta_{min} > 1$, we can just set $\beta_{min} = 1$ and the function remains well-behaved.

At a point $\bx$, if $M^{(\ell)}_{f}(\bx) \le C^\ell 2^{B} \frac{\ell!}{\beta^\ell}$ for some $\beta \ge \beta_{min}$ and all $\ell \in \bZ_{\ge 0}$, then for any $k \in \bZ_{\ge 0}$
\begin{align*}
    M^{(k)}_g(\bx) \le M^{(k + \kappa)}_f(\bx) \le C^{k + \kappa} 2^{B} \frac{(k + \kappa)!}{\beta^{k + \kappa}} \le C^{k} 2^{B + \kappa \lg(C) + \kappa \lg(1/\beta_{min})} \frac{(k + \kappa)!}{\beta^{k}}.
\end{align*}
Now, let us focus on $(k + \kappa)!$
\begin{align*}
    (k + \kappa)! &= k! (k + 1) (k + 2) \ldots (k + \kappa) \le k! (k + \kappa)^\kappa \le k! (k+1)^{\kappa}(\kappa+1)^{\kappa}\\
    &\le k! 2^{\kappa \lg(k+1)} 2^{\kappa \lg(\kappa + 1)}
    \le k! 2^{\kappa k} 2^{\kappa \lg(\kappa + 1)} = k! (2^{\kappa})^k 2^{\kappa \lg(\kappa + 1)}.
\end{align*}
Plugging this back into our bound on $M^{(k)}_g(\bx)$
\begin{align*}
    M^{(k)}_g(\bx) &\le C^{k} 2^{B + \kappa \lg(C) + \kappa \lg(1/\beta_{min})} \frac{k!}{\beta^{k}} (2^{\kappa})^k 2^{\kappa \lg(\kappa + 1)} \\
    &= (2^{\kappa} C)^{k} 2^{B + \kappa \lg(C) + \kappa \lg(1/\beta_{min}) + \kappa \lg(\kappa + 1)} \frac{k!}{\beta^{k}}.
\end{align*}
So, $g$ satisfies the well-behavedness condition with parameters $\widehat{B} = B + \kappa \lg(C) + \kappa \lg(1/\beta_{min}) + \kappa \lg(\kappa + 1)$ and $\widehat{C} = 2^\kappa C$. Since $B=\poly(|r|)$, $C = \poly(|r|)$, $\sz(\beta_{min}) = \poly(|r|)$, and $\kappa$ is constant, we have that $\widehat{B} = \poly(|r|)$ and $\widehat{C} = \poly(|r|)$, as desired. All other parameters ($(n_\beta)_\beta$, $\beta_{min}$) remain the same as for $f$.
\end{proof}

\begin{lemma}\label{lm:linear-comb}
For any well-behaved family $\mathcal{F} = (f_r)_r$, the family of all linear combinations
$$\mathcal{F}' = \left\{\sum_{i \in [n]} \alpha_i f_{r_i}: n \in \mathbb{N}, r_1, \dots, r_n \in \{0,1\}^n, \alpha_1, \dots, \alpha_n \in \mathbb{Q}\right\}$$
is also well-behaved. Here the domain of the function $\sum_{i \in [n]} \alpha_i f_{r_i}$ is taken to be $\bigcap_{i \in [n]} X_{r_i}$.
\end{lemma}
\begin{proof}
If a function $g = f_r$ satisfies the well-behavedness condition in $X$ then $\alpha g$, where $\alpha \in \mathbb{Q}$, also satisfies the condition in $X$ because: $M^{(k)}_{\alpha g} = |\alpha| M^{(k)}_{g}$, and we can incorporate the extra $|\alpha|$ factor into the ``$B$'' in \cref{def:well-behaved} by setting $B' = B + \max(0,\lg(|\alpha|))$. Note that since $B = \poly(|r|)$, we have $B' = \poly(|r|, \sz(\alpha))$, which is polynomial in the description of the function (namely, $|r|$ and $\sz(\alpha)$). As a result, the family of all possible scaled versions of functions in $\mathcal{F}$ is also a well-behaved family. It remains to show that taking sums of functions also retains the well-behavedness property.

Now, let $f_i = f_{r_i}$ be well-behaved in $X_i$ and let $f = \sum_i f_i$. Let $X = \bigcap_i X_{r_i}$; as $X_{r_i} \supseteq X$, $f_i$ is well-behaved in $X$. Fix a $\beta > 0$. Let $B_i$, $C_i$, $\beta_{i,min}$ and $n_{i,\beta}$ be the parameters associated with the well-behavedness property of the function $f_i$. 
We claim that $f$ satisfies the well-behavedness property with parameters $B = \lg(n) + \max_i B_i$, $C = \sum_i C_i$, $\beta_{min} = \min_i \beta_{i,min}$ and $n_{\beta} = \sum_i n_{i, \beta}$, where the $n_{\beta}$ hypercuboids associated with $f$ are simply the joint list of all the $n_{i,\beta}$ hypercuboids associated with each $f_i$ for all $i \in [n]$. Note that we have $B = \poly(|r_1|, \dots, |r_n|)$, $C = (|r_1|, \dots, |r_n|)$, $\sz(\beta_{min}) = \poly(|r_1|, \dots, |r_n|)$, and $n_\beta = \poly(|r_1|, \dots, |r_n|, \sz(\beta))$, as desired.

At a point $\bx$, if $M^{(k)}_{f_i}(\bx) \le C_i^k 2^{B_i} \frac{k!}{\beta^k}$ for all $i \in [n]$, then for $k \ge 1$ we have
\begin{align*}
    M^{(k)}_{f}(\bx) = M^{(k)}_{\sum_i f_i}(\bx) = \sum_i M^{(k)}_{f_i}(\bx) \le \sum_i C_i^k 2^{B_i} \frac{k!}{\beta^k} \le \sum_i C_i^k 2^{B} \frac{k!}{\beta^k} &\le (\sum_i C_i )^k 2^{B} \frac{k!}{\beta^k}\\
    &= C^k 2^{B} \frac{k!}{\beta^k}.
\end{align*}
On the other hand, for $k = 0$, we have $M^{(0)}_{f}(\bx) = \sum_i M^{(0)}_{f_i}(\bx) = \sum_i 2^{B_i} \le n 2^{\max_i B_i} = 2^B$.
So, if $ M^{(k)}_{f}(\bx) > C^k 2^{B} \frac{k!}{\beta^k}$ at a point $\bx$, then there must be an $i \in [n]$ such that $M^{(k)}_{f_i}(\bx) > C_i^k 2^{B_i} \frac{k!}{\beta^k}$, then $\bx$ is in one of the $n_{i,\beta}$ hypercuboids associated with $f_i$, and $f$ also has this hypercuboid.
\end{proof}

\begin{lemma}\label{lm:product-comb}
For any well-behaved family $\mathcal{F} = (f_r)_r$, the family of all products
$$\mathcal{F}' = \left\{\prod_{i \in [n]} f_{r_i}: n \in \mathbb{N}, r_1, \dots, r_n \in \{0,1\}^n\right\}$$
is also well-behaved. Here the domain of the function $\prod_{i \in [n]} f_{r_i}$ is taken to be $\bigcap_{i \in [n]} X_{r_i}$.
\end{lemma}
\begin{proof}
Note that $X_{r_i} \supseteq X$, so $f_i = f_{r_i}$ is well-behaved in $X$.
Let $B_i$, $C_i$, $\beta_{i,min}$, and $(n_{i,\beta})_{\beta}$ be the parameters associated with the well-behavedness property of the function $f_i$. 
We claim that $f = \prod_{i \in [n]} f_i$ is satisfies the well-behavedness property with parameters $B = \sum_i B_i$, $C = n \max_i C_i$, $\beta_{min} = \min_i \beta_{i,min}$, and $n_{\beta} = \sum_i n_{i, \beta}$, where the $n_{\beta}$ hypercuboids associated with $f$ are the union of the $n_{i,\beta}$ hypercuboids associated with each $f_i$ for all $i \in [n]$.
Note that we have $B = \poly(|r_1|, \dots, |r_n|)$, $C = (|r_1|, \dots, |r_n|)$, $\sz(\beta_{min}) = \poly(|r_1|, \dots, |r_n|)$, and $n_\beta = \poly(|r_1|, \dots, |r_n|, \sz(\beta))$, as desired.

At a point $\bx$, if $M^{(\ell)}_{f_i}(\bx) \le C_i^\ell 2^{B_i} \frac{\ell!}{\beta^\ell}$ for all $i \in [n]$ and $\ell \in \bZ_{\ge 0}$, then using the product rule of differentiation
\begin{align*}
    M^{(k)}_{f}(\bx) 
    &= M^{(k)}_{\prod_i f_i}(\bx) 
    = \sum_{|\bs| = k} \prod_{i \in [n]} M^{(s_i)}_{f_i}(\bx), \qquad \text{ where $\bs \in \bZ_{\ge 0}^n$ is a multi-index,}\\
    &\le \sum_{|\bs| = k} \prod_{i \in [n]} C_i^{s_i} 2^{B_i} \frac{s_i!}{\beta^{s_i}}
    \le (\max_i C_i)^k 2^{\sum_i B_i} \frac{1}{\beta^k} \sum_{|\bs| = k} \bs! \\
    &\le (\max_i C_i)^k 2^B \frac{1}{\beta^k} n^k k! = C^k 2^B \frac{k!}{\beta^k},
\end{align*}
where the last inequality holds because $\sum_{|\bs| = k} \bs! \le \sum_{|\bs| = k} k! \le n^k k!$.
So, if $ M^{(k)}_{f}(\bx) > C^k 2^{B} \frac{k!}{\beta^k}$ at a point $\bx$, then there must be an $i \in [n]$ such that $M^{(k)}_{f_i}(\bx) > C_i^k 2^{B_i} \frac{k!}{\beta^k}$, then $\bx$ is in one of the $n_{i,\beta}$ hypercuboids associated with $f_i$, and $f$ also has this hypercuboid.
\end{proof}

\begin{lemma}\label{lm:polynomial-well-behaved}
The family of all polynomials
$$\mathcal{F} = \left\{p: [-1,1]^d \to \mathbb{R}, \bx \mapsto \sum_{|\bs| \leq n} \alpha_{\bs} \bx^{\bs} : d \in \mathbb{N}, n \in \mathbb{N}, \alpha_{\bs} \in \mathbb{Q}\right\}$$
is well-behaved.\footnote{Here we slightly extend the definition of well-behavedness to allow functions on spaces of variable dimension~$d$.}
\end{lemma}

\begin{proof}
Given the list of all coefficients $\alpha_{\bs}$, $|\bs| \leq n$, of a polynomial $p$, we can in polynomial time construct a list of all the coefficients of all the derivatives $D^{\bs} p$ for all $\bs$. This is possible because the derivatives will be zero for any $|\bs| > n$. In particular, we can in polynomial time compute an upper bound $M$ on the value that any of these polynomials can take over $[-1,1]^d$. This means that $\sz(M)$ is polynomial in the size of the list of all coefficients of $p$.

Now, we claim that the family satisfies the well-behavedness property with $B = \lceil \log M \rceil$, $C = 2$, $\beta_{min} = 2$, and $n_\beta = 1$. Indeed, if $\beta \geq 2$, then the hypercuboid $[-1,1]^d$, which has width 2, covers the whole domain and there is nothing to show. Finally, for all $k \in \mathbb{Z}_{\geq 0}$ and all $\bx \in [-1,1]^d$ we have
$$M_p^{(k)}(\bx) \leq M \leq 2^B = C^k 2^B \frac{1}{2^k} \leq C^k 2^B \frac{k!}{\beta_{min}^k}.$$
Thus, $p$ satisfies the well-behavedness property.
\end{proof}

\section{Generalized Potentials: Existence and \tfnp/ Formulation}
\label{app:generalized-potentials}

For any $\delta > 0$ and $\bx \in \mathbb{R}^d$, let $B_\delta(\bx) = \{\by \in \mathbb{R}^d: \|\bx - \by\|_2 \leq \delta\}$. We let $\partial B_\delta(\bx) = \{\by \in \mathbb{R}^d: \|\bx - \by\|_2 = \delta\}$. Let $\langle \bx, \by \rangle$ denote the inner product.

\begin{definition}
A generalized potential in $d$-dimensional space is a continuously differentiable function $f: B_R(0) \setminus \cup_i B_r(\ba_i) \to \mathbb{R}$, where:
\begin{itemize}
    \item the points $\ba_1, \dots, \ba_n \in \mathbb{R}^d$ are distinct,
    \item $r > 0$ satisfies $B_r(\ba_i) \cap B_r(\ba_j) = \emptyset$ for any distinct $i$ and $j$,
    \item $R > 0$ satisfies $\cup_i B_{2r}(\ba_i) \subset B_R(0)$.
\end{itemize}
\end{definition}

\begin{definition}
Let $\delta > 0$. We say that a generalized potential $f: B_R(0) \setminus \cup_i B_r(\ba_i) \to \mathbb{R}$ satisfies the boundary conditions with signs $s_0, s_1, \dots, s_n \in \{-1,+1\}$ and with parameter $\delta$ if the following conditions hold:
\begin{itemize}
    \item For all $\bx \in \partial B_R(0)$, we have $s_0 \cdot \langle \nabla f (\bx), \bx \rangle \geq \delta$.
    \item For all $i \in [n]$ and all $\bx \in \partial B_r(\ba_i)$, we have $s_i \cdot \langle \nabla f (\bx), \bx - \ba_i \rangle \geq \delta$.
\end{itemize}
\end{definition}

\begin{theorem}\label{thm:existence}
Let $f$ be a generalized potential in $d$-dimensional space with $n \geq 2$ centers that satisfies the boundary conditions with signs $s_0, s_1, \dots, s_n$. Then $f$ admits an equilibrium point $\bx$, i.e., $\nabla f(\bx) = 0$, in the following cases:
\begin{itemize}
    \item The dimension $d$ is even.
    \item The dimension $d$ is odd and $s_0 - \sum_{i=1}^n s_i \neq 0$.
\end{itemize}
\end{theorem}

\begin{proof}[Proof sketch]
We use a Sperner-type argument.
Consider first the case where $d=2$. Pick some small $\eps > 0$ and an $\eps$-fine triangulation of the domain $B_R(0) \setminus \cup_i B_r(\ba_i)$. Color each vertex of the triangulation with one of three colors $1, 2,$ or $3$, depending on the direction in which $\nabla f(\bx)$ points (where the set of directions is divided equally and contiguously between the three colors in some arbitrary way). Next, examining the colors on the boundaries of the domain, we can make the following observations:
\begin{itemize}
    \item If we travel in a clock-wise manner on the boundary of the domain around a given center, then we will see the colors appear in a particular order,\footnote{To be more precise, when the color switches from, say, $1$ to $2$, it might momentarily switch back to color $1$, before ``stabilizing'' to color $2$. This is a technical detail that does not impact the general argument.} say color $1$, then $2$, then $3$.
    \item Importantly, the order in which the colors are seen when traveling clock-wise is the same for all centers, no matter their sign.
    \item Furthermore, if we travel clock-wise on the outer boundary of the domain, we again see the colors in that same order.
\end{itemize}
Now, using standard Sperner-type arguments (e.g., the proof of Sperner's lemma that uses a path-following argument), one can argue that the triangulation must contain a triangle with all three colors. Then, using standard continuity and limit arguments by taking $\eps \to 0$, one can show that an exact equilibrium point must exist.

The same argument can be generalized to also work for higher even dimensions.
When the dimension is odd, the boundary conditions are no longer identical around centers (namely, they depend on the sign), and an additional condition is required to ensure the existence of a trichromatic triangle. This condition boils down to the aforementioned condition on the signs.
\end{proof}

\begin{corollary}
An electrostatic potential in $d$-dimensional space with $n \geq 2$ charges with magnitudes $q_1, \dots, q_n \in \mathbb{R}_{\neq 0}$ admits an equilibrium point in the following cases:
\begin{itemize}
    \item The dimension $d$ is even and $\sum_i q_i \neq 0$.
    \item The dimension $d$ is odd and either
    \begin{enumerate}
        \item $\sum_i q_i > 0$ and $n_+ \neq n_- + 1$, or
        \item $\sum_i q_i < 0$ and $n_- \neq n_+ + 1$,
    \end{enumerate}
    where $n_+$ and $n_-$ denote the number of positive and negative charges, respectively.
\end{itemize}
\end{corollary}

\begin{proof}
In order to apply \cref{thm:existence}, we have to show that the electrostatic potential satisfies the conditions of the theorem. First of all, if $\sum_i q_i > 0$, there exists a sufficiently large $R > 0$ such that the boundary condition is satisfied with $s_0 = -1$. Similarly, if $\sum_i q_i < 0$, the boundary condition is satisfied with $s_0 = +1$. Next, there exists a sufficiently small $r > 0$ such that the boundary conditions are satisfied around each singularity, where $s_i$ will depend on the sign of the singularity, namely $s_i = - \textup{sgn}(q_i)$. The corollary then follows.
\end{proof}

\paragraph{\bf Computational \tfnp/ Formulation.}
We can define a computational problem where we are given $R, r, \delta$, the points $\ba_1, \dots, \ba_n \in \mathbb{R}^d$ and the signs $s_0, s_1, \dots, s_n \in \{-1,+1\}$, as well as a circuit computing a continuously differentiable potential $f: B_R(0) \setminus \cup_i B_r(\ba_i) \to \mathbb{R}$ satisfying\footnote{One can either consider a promise problem, where there is a promise that these conditions are satisfied, or introduce so-called violation solutions, that can be returned when a violation is detected.} the boundary conditions with these signs and with parameter $\delta$. Then, given an $\eps > 0$, the goal is to output $\bx \in B_R(0) \setminus \cup_i B_r(\ba_i)$ with $\|\nabla f(\bx)\| \leq \eps$.

For this problem, when the conditions for existence of \cref{thm:existence} are satisfied, one can show membership in \ppad/. This follows from our proof sketch of \cref{thm:existence} which uses a Sperner-type argument together with the results about the End-of-Line problem with multiple sources and sinks in~\cite{GoldbergH21}.

\bibliographystyle{alpha}
\bibliography{ref}

\end{document}